\documentclass[letterpaper,twocolumn,10pt]{article}
\usepackage{usenix}

\usepackage{tikz}
\usepackage{amsmath,amsfonts,amsthm,amssymb}
\usepackage{algorithm}
\usepackage[noend]{algpseudocode}
\usepackage{graphicx}
\usepackage{subcaption}
\usepackage{textcomp}
\usepackage{xcolor}
\usepackage{array}

\theoremstyle{plain}
\newtheorem{thm}{Theorem}
\newtheorem{prop}{Proposition}
\newtheorem{lem}[thm]{Lemma}

\theoremstyle{definition}
\newtheorem{dfn}{Definition}

\DeclareMathOperator*{\argmin}{argmin}

\DeclareMathOperator{\Var}{Var}

\begin{document}

\date{}

\title{\Large \bf N-output Mechanism: Estimating Statistical Information from Numerical Data under Local Differential Privacy}

\author{
{\rm Incheol \ Baek}\\
Korea University
\and
{\rm Yon Dohn Chung}\\
Korea University
} 

\maketitle

\begin{abstract}
Local Differential Privacy (LDP) addresses significant privacy concerns in sensitive data collection. In this work, we focus on numerical data collection under LDP, targeting a significant gap in the literature: existing LDP mechanisms are optimized for either a very small ($|\Omega| \in \{2, 3\}$) or infinite output spaces. However, no generalized method for constructing an optimal mechanism for an arbitrary output size $N$ exists. To fill this gap, we propose the \textbf{N-output mechanism}, a generalized framework that maps numerical data to one of $N$ discrete outputs. 

We formulate the mechanism's design as an optimization problem to minimize estimation variance for any given $N \geq 2$ and develop both numerical and analytical solutions. This results in a mechanism that is highly accurate and adaptive, as its design is determined by solving an optimization problem for any chosen $N$.  Furthermore, we extend our framework and existing mechanisms to the task of distribution estimation. Empirical evaluations show that the N-output mechanism achieves state-of-the-art accuracy for mean, variance, and distribution estimation with small communication costs.
\end{abstract}

\section{Introduction}

In today's interconnected world, collecting and analyzing data from various devices, such as computers, phones, and IoT devices, has become increasingly important. However, these data often contain sensitive information, raising significant privacy concerns. Local Differential Privacy (LDP)~\cite{kasiviswanathan2011can} has emerged as a solution to address these concerns. In LDP environments, clients perturb their data locally before sending them to a server (aggregator). This process ensures that the true data remain hidden from the server, which only has access to the perturbed data. Despite this, the server can still estimate statistical information, such as means, variances, and distributions, from the perturbed data. Given these benefits, big tech companies have adopted LDP techniques: Google in its Chrome browser \cite{erlingsson2014rappor}, Apple for emoji usage analysis~\cite{Apple}, and Microsoft for telemetry data collection~\cite{ding2017collecting}.

An LDP mechanism, formally denoted as $\mathcal{M}: \mathcal{X} \to \Omega$, is a function that maps a client's true value to a perturbed values in an output space $\Omega$. LDP mechanisms have been developed specifically for the type of data being collected, which defines the input space $\mathcal{X}$. Much of the initial research focused on categorical data, where the goal is to estimate the frequency of various categories~\cite{wang2017locally, erlingsson2014rappor, bassily2015local, wang2019consistent}. While it is possible to handle numerical data by discretizing them into categories, this approach often fails to preserve inherent properties of numerical data. To overcome this limitation, subsequent research has focused on developing mechanisms specifically for numerical data collection~\cite{duchi2018minimax, wang2019collecting, zhao2020local, li2020estimating}, which can accurately estimate statistics like mean and distribution while providing rigorous privacy guarantees.

Existing numerical LDP mechanisms can be characterized by the size of their output space, $|\Omega|$. For example, Duchi's mechanism~\cite{duchi2018minimax} randomly maps a value to one of two possible outputs ($|\Omega|=2$), while the Three-output mechanism~\cite{zhao2020local} uses three ($|\Omega|=3$). In contrast, others like Piecewise Mechanisms (PM)~\cite{zhao2020local,wang2019collecting} map to a continuous range ($|\Omega|=\infty$). The size of this output space is a critical design choice, especially concerning the privacy budget, $\epsilon$. A smaller $\epsilon$ provides stronger privacy but requires adding more noise, which can harm data utility. For small $\epsilon$ (in the high-privacy regime), it is known that mechanisms with a smaller number of outputs tend to offer better utility ~\cite{zhao2020local, wang2019collecting}.

While optimized mechanisms have been developed for $|\Omega|= 2, 3, \infty$, there has been limited exploration of mechanisms designed for a finite number of outputs. This remains an problem in the field. Motivated by this gap, we propose a novel and generalized LDP mechanism called the \textbf{N-output mechanism}, which maps an input value to one of $N$ discrete outputs. We formulate a general construction method that optimizes the mechanism's utility for given number of outputs, $N \geq 2$. This allows our mechanism to achieve high performance for various statistical estimation tasks---such as mean, variance, and distribution estimation---across a wide spectrum of privacy settings.

Furthermore, the N-output mechanism offers several key advantages. First, it is bit-efficient, reducing communication costs compared to PMs~\cite{zhao2020local, wang2019collecting} that rely on float-point representations. Second, by allowing for a sufficiently large yet finite $N$, it enables accurate distribution estimation, a task that is challenging for mechanisms with a minimal number of outputs like Duchi's~\cite{duchi2018minimax} or the Three-output mechanism~\cite{zhao2020local}. Consequently, our proposed mechanism not only fills the existing gap but also generalizes prior work into a single, flexible, and powerful framework.

Our contributions are as follows:
\begin{itemize}
    \item \textbf{A novel, generalized mechanism.} We introduce the \textbf{N-output mechanism}, a new LDP framework for numerical data collection. We formulate its design as a solvable optimization problem to find the optimal mechanism for any given number of outputs, $N \geq 2$.
  
\item \textbf{Dual optimization solutions.} We provide both a highly accurate numerical solution that achieves state-of-the-art (SOTA) utility and a theoretically-grounded analytical solution that generalizes prior works~\cite{duchi2018minimax, zhao2020local}.

\item \textbf{Advancing distribution estimation.} We enable distribution estimation for numerical LDP mechanisms, which has not been fully addressed. We perform LDP distribution estimation~\cite{li2020estimating} using the Expectation-Maximization (EM) algorithm~\cite{dempster1977maximum}, for which we also derive the required general probability transition matrices.

\item \textbf{Empirical evaluation.} We empirically demonstrate that our mechanism achieves superior accuracy for mean, variance, and distribution estimation across a wide range of privacy settings, while being up to 16 times more communication-efficient than baselines.
\end{itemize}

\section{Related Work}
LDP has been extensively studied for different data types and estimation tasks. In this section, we review the foundational work in both categorical and numerical data collection under LDP.

\subsection{LDP for categorical data}
The pioneering work in LDP largely focused on frequency estimation for categorical data. A foundational technique is Randomized Response (RR)~\cite{warner1965randomized}, where users respond truthfully to a question with a certain probability. Direct Encoding (DE)~\cite{wang2017locally}, a generalized randomized response, is the most straightforward approach where the categorical value itself is perturbed and reported. 

Another widely adopted approach is Unary Encoding (UE), also known as one-hot encoding. In UE, a categorical value $v$ from a domain of size $d$ is encoded into a $d$-dimensional bit vector where only the $v$-th bit is 1. The perturbation then happens on this vector, where each bit is flipped with specific probabilities ($p$ for a 1-bit, $q$ for a 0-bit). Symmetric Unary Encoding (SUE), Google's RAPPOR~\cite{erlingsson2014rappor}, uses symmetric flipping probabilities where $p + q = 1$. Optimized Unary Encoding (OUE)~\cite{wang2017locally} improves upon SUE by selecting $p$ and $q$ to specifically minimize the estimation variance. 

While effective for their intended purpose, these methods are not directly suitable for numerical data. Although these categorical mechanisms can be adapted for numerical data through discretization, this process often discards numerical characteristics, motivating the need for specialized numerical mechanisms.

\subsection{LDP for numerical data}
Duchi's Mechanism~\cite{duchi2018minimax} is a foundational work for mean estimation over a continuous domain such as $[-1, 1]$. It is characterized by an extremely small output space, mapping any input value to one of only two possible outputs $(|\Omega|=2)$. While highly communication-efficient and effective in high privacy (small $\epsilon$) regimes, its minimal output space makes it unsuitable for more complex tasks like distribution estimation and leads to a significant drop in accuracy in low privacy regimes.

The Three-output mechanism~\cite{zhao2020local} was proposed as an extension, increasing the output space to three discrete values $(|\Omega|=3)$. This addition allows better utility in certain privacy regimes compared to Duchi's mechanism~\cite{duchi2018minimax}. However, it still faces significant limitations in accuracy and applicability to complex tasks due to its small output space.

PM~\cite{wang2019collecting} represents a different approach, mapping a numerical input to a continuous output range $(|\Omega|=\infty)$. Later work by~\cite{zhao2020local} further optimized PM into sub-optimal (PM-sub) version. PMs allow for high utility, especially in low privacy regimes. However, this comes at the cost of higher communication overhead, as transmitting a continuous (floating-point) value requires significantly more bits than transmitting a discrete value.

Our proposed N-output mechanism is motivated by the gap between these two extremes: the highly constrained, discrete outputs of Duchi's~\cite{duchi2018minimax} and the Three-output mechanisms~\cite{zhao2020local}, and the infinite, continuous output of PMs~\cite{zhao2020local,wang2019collecting}. Our objective is to achieve high utility for various complex tasks while maintaining low communication costs.

\section{Preliminaries}
\subsection{Local Differential Privacy}
 LDP stems from Differential Privacy (DP)~\cite{dwork2006differential}, a widely recognized standard for preserving privacy in data publishing and analysis. In DP models, a data curator releases statistical information and analysis, perturbed by a randomized algorithm, preventing adversaries from inferring individuals' information. However, this model has a vulnerability: the data curator initially has access to raw data, posing a risk if the data publisher is untrustworthy or if adversaries gain access to the raw data. LDP, defined in Definition \ref{def:LDP}, addresses this risk by having each client perturb their data before sending it to the aggregator.

\begin{dfn} \label{def:LDP}
A randomized mechanism $\mathcal{M}: \mathcal{X} \to \Omega$ satisfies $\epsilon$-LDP if and only if, for any two input values $x, x' \in \mathcal{X}$ and an output $y \in \Omega$, the following inequality holds:
\begin{equation*}
    \Pr[\mathcal{M}(x) = y] \leq e^{\epsilon}\Pr[\mathcal{M}(x') = y]
\end{equation*}
\end{dfn}

For the additional preliminaries for Duchi's~\cite{duchi2018minimax}, three-output~\cite{zhao2020local}, and PMs~\cite{wang2019collecting, zhao2020local} are in Appendix~\ref{sec:add_prelim}.

\subsection{Mean estimation under LDP}
Let the input domain be $\mathcal{X} = [-1, 1]$ and the output space be $\Omega$. Let $x \in \mathcal{X}$ be the true value of a user, and $y \in \Omega$ be the perturbed value produced by an LDP mechanism. Let $X$ and $Y$ be the random variables representing the true value and the perturbed value, respectively. The probability of reporting $y$ given that the true value is $x$ is denoted as $\Pr[Y=y|X=x]$, which we simplify as $\Pr[y|x]$. The expected value and variance of the perturbed value $Y$ given the true value $x$ are denoted as $\mathbb{E}[Y|x]$ and $\text{Var}[Y|x]$, respectively. For correct mean estimation, the unbiasedness constraint, $\mathbb{E}[Y|x] = x$ must hold.

\section{N-output Mechanism}
In this section, we introduce the \textbf{N-output mechanism}. We begin by defining the general optimization problem that numerical LDP mechanisms~\cite{duchi2018minimax, zhao2020local, wang2019collecting} aim to solve. We then formulate this problem for a discrete output space of size $N$, transforming an intractable problem into a solvable form. Subsequently, we present both a numerical solution that achieves high accuracy and an analytical solution that, while simplified, offers theoretical reliability and generalizes previous works. Lastly, we provide a comprehensive analysis of our mechanism and compare its performance with existing mechanisms.

\subsection{Problem definition}
A client has a value $x \in [-1, 1]$ and the randomly perturbed value is $y = \mathcal{M}(x)$, where $y \in \Omega$. For simplicity, we assume that each client send a scalar value. The goal is to find $\mathcal{M}$ minimizing a cost function subject to certain constraints. The cost function is typically defined using the noise variance, $\text{Var}[Y|x]$. Since numerical LDP mechanisms require the unbiasedness constraint $\mathbb{E}[Y|x]=x$, the variance becomes equivalent to the Mean Squared Error (MSE):
\begin{equation}
    \Var[Y|x] = \int_{\Omega} (y-x)^2 \Pr[y|x] dy
\end{equation}
Thus, the noise variance directly measures how far the perturbed value can deviate form the original value. To evaluate the overall performance of the mechanism across all inputs, we define the cost function using the $r$-norm of the noise variance. This provides a generalized measure of the error. Let $f(x)$ be the probability distribution of the input $x$. The cost function to be minimized is:
\begin{equation}
    J_\mathcal{M}^{(r)} = \left( \int_{\mathcal{X}} (\Var[Y|x])^r f(x) dx \right)^{1/r}
\end{equation}
Our goal is to find $\mathcal{M}$ that minimizes this cost function while satisfying LDP and other probability constraints:
\begin{equation} \label{LDP_problem}
\begin{aligned}
\textrm{Minimize} \quad &  \boldsymbol{\left( \int_{\mathcal{X}} \left( \text{Var}[Y|x] \right)^r f(x) dx \right)^{1/r}}\\
\textrm{subject to} \quad & \text{Pr}[y | x] \leq e^{\epsilon}\text{Pr}[y | x']\\
  & \mathbb{E}[Y|x] = x \\
  & \int_{\Omega} \text{Pr}[y | x] dy = 1 \\
  & \text{Pr}[y | x] \geq 0 \\
\end{aligned}
\end{equation}

\begin{figure*}[ht!]
    \centering
    \begin{subfigure}[b]{0.495\textwidth}
        \centering
        \includegraphics[width=\textwidth]{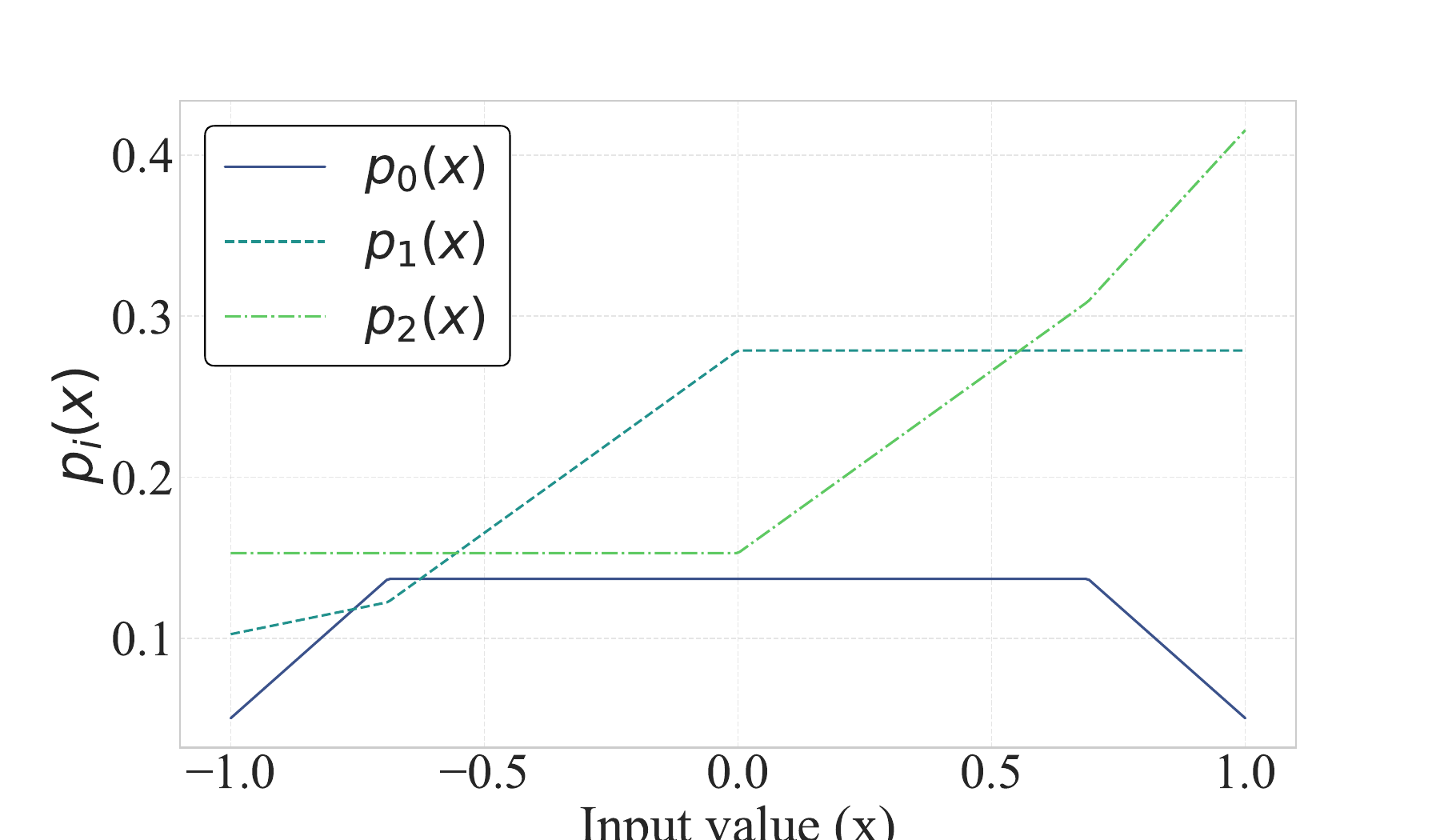} 
        \caption{Probability functions $p_i(x)$}
        \label{fig:pr_ai}
    \end{subfigure}
    \begin{subfigure}[b]{0.495\textwidth}
        \centering
        \includegraphics[width=\textwidth]{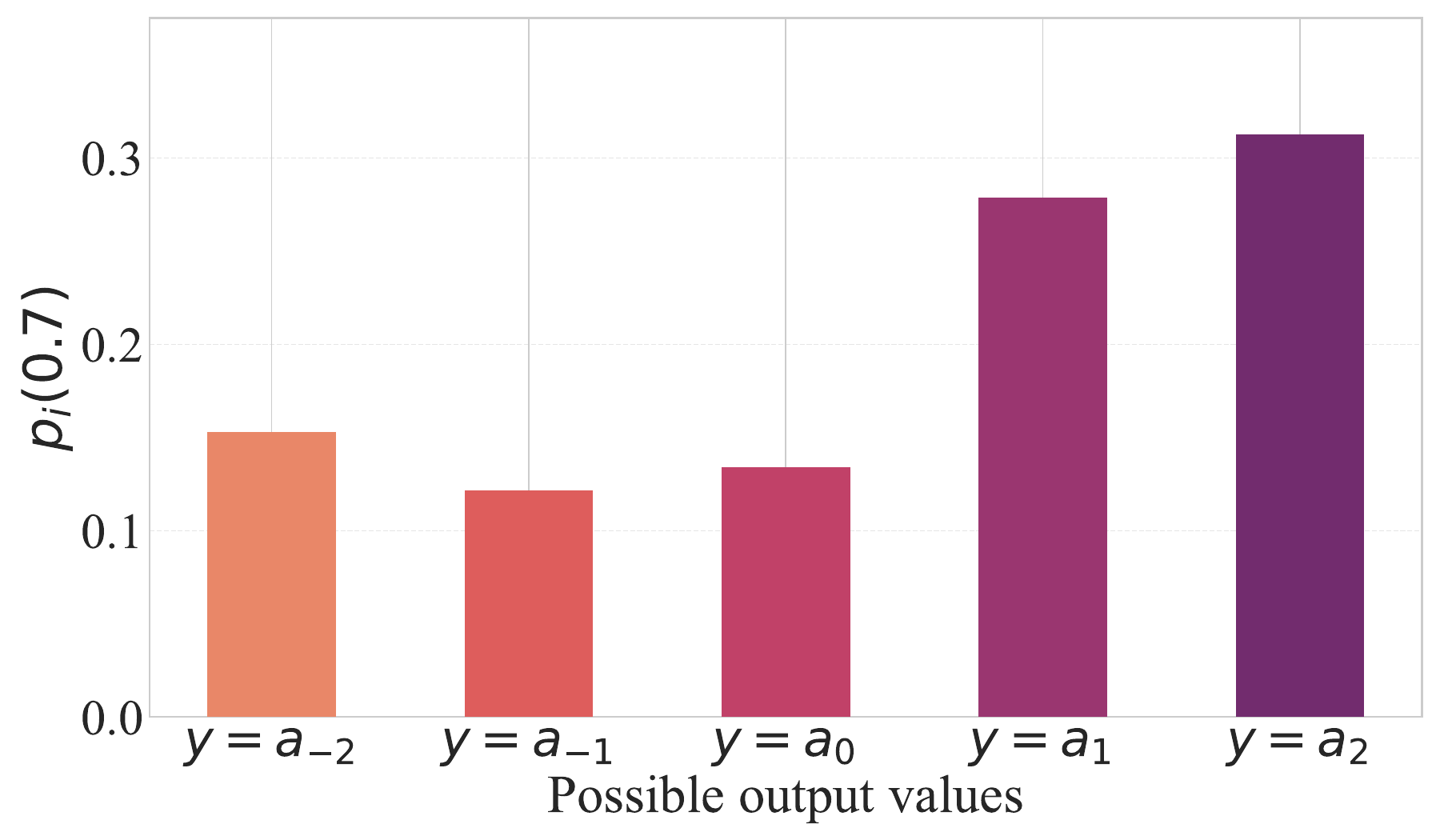} 
        \caption{Output probability distribution for input $x=0.7$.}
        \label{fig:pr_y}
    \end{subfigure}
    \caption{An example of the N-output mechanism for $N=5$ and $\epsilon=1$. The output values are $a_1 \approx 1.87$ and $a_2 \approx 2.57$. The symmetric outputs $a_{-1}, a_{-2}$ are omitted for clarity in (a).}
    \label{fig:pr}
\end{figure*}

\subsection{Problem formulation}
In this work, our goal is to find an adaptive mechanism $\mathcal{M}$ that optimizes problem \eqref{LDP_problem} for a given output size $|\Omega|=N$. We define a symmetric, discrete output space $\Omega$ as follows:
\begin{gather*}
\Omega= \begin{cases}
    \{ a_i\}_{i=-n}^{n} \setminus \{a_0\} & \text{if } N = 2n \\ 
    \{ a_i\}_{i=-n}^{n} & \text{if } N = 2n + 1 \\ 
\end{cases}
\\
\text{where } a_{-i} = -a_i, \quad a_0 = 0,  \quad \text{and } a_i < a_{i+1} \text{ for all } i.
\end{gather*}
A discrete-output LDP mechanism $\mathcal{M}$ maps an input $x \in \mathcal{X}$ to an output $a_i \in \Omega$ with a randomization probability, which we denote as $p_i(x) \triangleq \Pr[a_i|x]$. Our objective is to determine the output values $\{a_i\}_{i=-n}^{i=n}$ and the probability functions $\{p_i(x) \}_{i=-n}^{i=n}$ that minimize the cost function while satisfying $\epsilon$-LDP. 

However, finding the optimal functions $p_i(x)$ for all $x \in [-1, 1]$ is an infinite-dimensional optimization problem and thus computationally intractable. To make the problem solvable, we model each $p_i(x)$ as a piecewise linear functions, which is flexible enough to approximate complex distributions and also adopted in the previous works~\cite{duchi2018minimax, zhao2020local}. This reduces the problem to a finite-dimensional optimization task, as we only need to determine the probability values at a finite number of points.

Specifically, we divide the interval $[-1, 1]$ into sub-intervals $[x_{j-1}, x_j]$. The restriction of $p_i(x)$ to any such sub-interval is a single linear function, which we denote by $p_{i,j}(x)$. In other words, $p_i(x) = p_{i,j}(x)$ for $x \in [x_{j-1}, x_j]$. This linear function $p_{i,j}(x)$ is defined by the interpolation between its values at the endpoints. We denote the probability value $p_i(x)$ at an endpoint $x_j$ as $P_{i,j}$. The function $p_{i,j}(x)$ for $x \in [x_{j-1} , x_j]$ is then given by for $i=-n, \ldots, n$ and $j=-n+1, \ldots, n$:
\begin{equation}\label{pij}
    p_{i,j}(x) = \frac{P_{i, j} - P_{i, j-1} }{x_j - x_{j-1}} (x - x_{j}) + P_{i,j}
\end{equation}

The variance of the output $Y$ for a given input $x$ is defined as:
\begin{equation}
    \Var[Y|x] = \sum_{i=-n}^n a_i^2 p_i(x) - x^2
\end{equation}

To illustrate the behavior of the N-output mechanism, Figure~\ref{fig:pr} shows an example of our mechanism with $N=5$ and $\epsilon=1$. Figure~\ref{fig:pr_ai} plots the probability $p_i(x)$. As shown, the probability of selecting a particular output $a_i$ changes with $x$, and an input is most likely to be mapped to the nearest output value. For a specific instance, Figure~\ref{fig:pr_y} shows the probability distribution for the input $x=0.7$.

\textbf{Challenges.} While this formulation is tractable, its primary challenge lies in the high dimensionality of the optimization problem. The optimization variables ($P_{i,j}$ and $a_i$) are subject to several equality constraints given in \eqref{LDP_problem}. The number of variables not fixed by these constraints---the free variables---determines the complexity of the optimization.

Specifically, the number of free variables for determining $P_{i,j}$ grows quadratically with $N$. If $N=2n$, there are $n(N-2)$ free variables, and if $N=2n+1$, there are $n(N - 1) - 1$ free variables. This complexity explains why prior works were limited to simple cases. For example, Duchi's mechanism ($N=2$) \cite{duchi2018minimax} has no such variables, while the three-output mechanism ($N=3$) \cite{zhao2020local} has only one. In contrast, the N-output mechanism requires optimizing $\mathcal{O}(N^2)$.

Furthermore, there is an inherent interdependence between $p_{i, j}(x)$ and $a_i$. Adjusting one set of variables invariably affects the other, making it infeasible to optimize them independently. This interdependence adds a layer of nonlinearity and non-convexity to the optimization problem.

In summary, the combination of non-convexity, non-linearity, and a large number of interdependent variables makes the optimization problem hard. Thus, we numerically solve the optimization problem using Sequential Least Squares Programming (SLSQP) algorithm. Moreover, we find that by appending certain constraints, we can simplify the optimization problem to a form that not only generalizes existing mechanisms but also allows for an analytical solution. We present the numerical solution first, followed by this analytical approach for the simplified problem.

\subsection{Numerical solution}
To solve the optimization problem numerically, we employ the SLSQP algorithm, designed to solve non-linear programming problems with both equality and inequality constraints. At each iteration, it approximates the original problem with a quadratic subproblem to find an optimal search direction.

\textbf{Variables.} The optimization variables involve the output space $\Omega = \{a_i\}_{i=-n}^n$ and the set of endpoint probabilities $\mathcal{P} = \{P_{i,j}\}$ for $i,j \in \{-n,\ldots, n\}$. Especially, $N=2n$, $P_{0,j}=0$ fixed.  In the following, we detail the formulation for an odd number of outputs ($N=2n+1$). Note that the endpoints $x_j$ are not variables to optimize. By unbiasedness $x_j$ is defined as,
\begin{equation}
    x_j = \sum_{i=-n}^n a_i P_{i,j}
\end{equation}

\textbf{Initialization.} A well-chosen initial point is crucial for stable convergence and finding a high-quality solution. We initialize our variables by adapting PM-sub~\cite{zhao2020local}. Specifically, we discretize the continuous output of PM-sub~\cite{zhao2020local} into $N$ values and derive the corresponding probabilities, providing a strong starting point for the optimization.

\textbf{Objective functions.} The objective function is the $r$-norm of the noise variance. $r=1$, it becomes the average-case noise variance, and for $r=\infty$, it becomes the worst-case noise variance. Existing works~\cite{wang2019collecting, duchi2018minimax, zhao2020local} mainly focus on optimizing the worst-case noise variance. Thus, we first address the $r \to \infty$ case and then the $r=1$ case.

As $r \to \infty$, the cost function becomes the maximum variance over the entire domain, $\max_{x \in [-1,1]}\Var[Y|x]$. However, it is computationally intractable to check the variance at every point $x \in [-1, 1]$. The following proposition states that the maximum variance within each piecewise linear segment $[x_{j-1}, x_j]$ can only occur at one of three candidate points.
\begin{prop} \label{prop:1}
For $[x_{j-1}, x_j]$, its maximum value within this interval must occur at one of the endpoints, $x_{j-1}$ or $x_j$, or at its vertex, $\bar{x}_j$, if the vertex lies within the interval. The vertex $\bar{x}_j$ is defined as:
\begin{equation}
    \bar{x}_j = \frac{\sum_{i=-n}^n (P_{i,j} - P_{i,j-1}) a_i^2}{2\sum_{k=-n}^n (P_{k,j} - P_{k, j-1})a_k}
\end{equation}
    \begin{proof}
        The proof is provided in the Appendix~\ref{pf:prop:1}.
    \end{proof}
\end{prop}
Therefore, the objective function can be reformulated into a tractable problem by checking only a finite set of points:
\begin{equation}
    \underset{\Omega, \mathcal{P}}{\text{Minimize}} \quad \max_{j \in \{-n+1, \ldots, n\}, x \in \{x_{j-1}, \bar{x}_j, x_j\}} \Var[Y|x]
\end{equation}

\textbf{Constraints.} There are four types of constraints: unbiasedness, LDP, and probability validity. An excessive number of constraints can increase computational cost and lead to convergence issues. Therefore, we simplify them where possible without loss of generality.

\textbf{Unbiasedness.} For the mechanism to be unbiased, the expected value must equal the input, $\mathbb{E}[Y|x] = x$. Our piecewise linear formulation is designed to inherently satisfy this condition, as stated in the following proposition.
\begin{prop} \label{prop:2}
Let the following conditions hold:
\begin{enumerate}
    \item Strictly ordered outputs: $a_{i-1} < a_i$ \, $\forall i$.
    \item Appropriate domain boundaries: $x_{-n} = -1$, $x_n = 1$.
\end{enumerate}
Then, for any input $x \in [-1, 1]$, the mechanism is unbiased:
\begin{equation}
    \mathbb{E}[Y|x] = x
\end{equation}

\begin{proof}
    The proof is provided in the Appendix~\ref{pf:prop:2}.
\end{proof}
\end{prop}
Therefore, to guarantee unbiasedness, the following constraints are enforced in the optimization problem:
\begin{gather}
    a_{i-1} < a_i \quad \forall i \\
    \sum_{i=-n}^n a_i P_{i,-n} = -1 , \quad \sum_{i=-n}^n a_i P_{i,n} = 1
\end{gather}

\textbf{LDP constraint.} A direct implementation of the LDP constraint, $p_i(x) \leq e^\epsilon \min_{x'} p_i(x')$, is problematic because the $\min$ operator is non-differentiable and it could fail to detect violation of the constraint. Fortunately, the structure of the probability functions simplifies this, as shown in the following proposition.
\begin{prop} \label{prop:3}
    For any output $a_i$, the minimum probability of $p_i(x)$ occurs at an endpoint of the domain, as follows:
    \begin{enumerate}
        \item If $i < 0$, then $\min_{x \in [-1, 1]} p_i(x) = p_i(1)$.
        \item If $i > 0$, then $\min_{x \in [-1, 1]} p_i(x) = p_i(-1)$.
        \item If $i = 0$, then $\min_{x \in [-1, 1]} p_i(x) = p_i(-1) = p_i(1)$.
    \end{enumerate}
    \begin{proof}
        The proof is provided in Appendix~\ref{pf:prop:3}.
    \end{proof}
\end{prop}
\noindent Thus, the LDP constraints is reformulated into:
\begin{gather}
    P_{i,-n} \leq P_{i,j} \leq e^\epsilon P_{i, -n} \quad  i \geq 0 \quad   \forall j\in \{-n, \ldots, n\}\\
        P_{i,n} \leq P_{i,j} \leq e^\epsilon P_{i, n} \quad  i<0 \quad   \forall j\in \{-n, \ldots, n\}
\end{gather}
Satisfying these constraints, the numerical solution is $\epsilon$-LDP.

\textbf{Probability validity.} To ensure a valid probability distribution, the output probabilities must be non-negative and sum to one. The following proposition formally states that it is sufficient to enforce these properties at the endpoints $x_j$.
\begin{prop} \label{prop:4}
    Let the probabilities at the endpoints satisfy the validity conditions for all $j \in \{-n, \ldots, n\}$:
    \begin{enumerate}
        \item \textbf{Non-negativity:} $P_{i,j} \geq 0$ for all $i$.
        \item \textbf{Sum-to-one:} $\sum_{i=-n}^n P_{i,j} = 1$.
    \end{enumerate}
    Then, for any input $x \in [-1, 1]$, the interpolated probabilities $p_i(x)$ also satisfy these conditions:
    \begin{enumerate}
        \item \textbf{Non-negativity:} $p_{i}(x) \geq 0$ \, $\forall i$.
        \item \textbf{Sum-to-one:} $\sum_{i=-n}^n p_{i}(x) = 1$.
    \end{enumerate}
    \begin{proof}
        The proof is provided in Appendix~\ref{pf:prop:4}.
    \end{proof}
\end{prop}

Based on the constraints, the final optimization problem is formulated as follows:
\begin{equation}
\begin{split}
& \underset{\Omega, \mathcal{P}}{\text{Minimize}} \quad \max_{j \in \{-n+1, \ldots, n\}, x \in \{x_{j-1}, \bar{x}_j, x_j\}} \Var[Y|x] \\
& \text{subject to} \\
& \qquad \sum_{i=-n}^n a_i P_{i,-n} = -1 , \quad \sum_{i=-n}^n a_i P_{i,n} = 1 \\
& \qquad a_{i-1} < a_i, \quad \forall i \in \{-n+1, \ldots, n\} \\
& \qquad \sum_{i=-n}^{n} P_{i,j} = 1, \quad \forall j \in \{-n, \ldots, n\} \\
& \qquad P_{i,-n} \leq P_{i,j} \leq e^\epsilon P_{i, -n} \quad  i \geq 0 \quad   \forall j\in \{-n, \ldots, n\} \\
& \qquad P_{i,n} \leq P_{i,j} \leq e^\epsilon P_{i, n} \quad  i<0 \quad   \forall j\in \{-n, \ldots, n\} \\
& \qquad P_{i,j} \geq 0, \quad \forall i, j \in \{-n, \ldots, n\}
\end{split}
\end{equation}

\textbf{Average-case noise variance.}
When $r=1$, the objective becomes the minimizing the average-case noise variance. While the worst-case objective provides a robust guarantee against the single worst output (a minimax approach from game theory), the average-case objective optimizes for overall performance across all possible outcomes. This can be powerful when maximizing the mechanism's expected utility. 

The average-case noise variance is the expected variance over the distribution of the input data, $f(x)$:
\begin{equation} \label{eq:avg_var}
    \mathbb{E}_{X \sim f(x)}[\Var[Y|X]] = \int_{-1}^{1} \Var[Y|x] f(x) dx
\end{equation}
As the true data distribution $f(x)$ is typically unknown, we assume a uniform distribution over the domain $[-1, 1]$, where the probability density function is $f(x)=1/2$. Under this assumption, the final objective function to minimize is:
\begin{equation}
\frac{1}{2} \sum_{j=-n+1}^{n} \int_{x_{j-1}}^{x_j} \Var[Y|x] dx
\end{equation}
where the integral over each segment resolves to:
\begin{equation} \label{eq:avg_var_int}
\begin{aligned}
    &\int_{x_{j-1}}^{x_j} \Var[Y|x] dx \\
    &= \frac{1}{2}\left(\sum_{i=-n}^n a_i^2 (P_{i,j-1} + P_{i,j})\right)(x_j - x_{j-1}) - \frac{1}{3}(x_j^3 - x_{j-1}^3)
\end{aligned}
\end{equation}
By substituting the result of Eq.~\eqref{eq:avg_var_int} into the summation, the final objective function for the average-case variance is:
\begin{equation}
\underset{\Omega, \mathcal{P}}{\text{Minimize}} \quad \sum_{j=-n+1}^{n} \left(\sum_{i=-n}^n a_i^2 (P_{i,j-1} + P_{i,j})\right)(x_j - x_{j-1})
\end{equation}
\noindent This objective is minimized subject to the same set of constraints defined in the worst-case optimization problem.

\subsection{Analytical solution}
A general analytical solution for the N-output mechanism is often intractable due to the complex interdependence of the variables. We address this by imposing a simplifying structure on the probability variables $\mathcal{P}$, removing the interdependencies. The simplified N-output mechanism still generalizes Duchi's~\cite{duchi2018minimax} and Three-output~\cite{zhao2020local} mechanisms and have low worst-case noise variance.

\textbf{Simplification.} Our approach is to pre-define the structure of the probabilities $P_{i,j}$. This structure is designed to satisfy the $\epsilon$-LDP constraint while concentrating probability mass on outputs $a_i$ close to the input $x$, which is crucial for minimizing the noise variance. Specifically, we set the probabilities as follows:
First, for $i = 0$, we define:
\begin{equation}
    P_{i,j} = \begin{cases}
    e^\epsilon p_0 &\quad \text{ if } i = j \\
    p_0  &\quad \text{otherwise}
    \end{cases}
\end{equation}
where $p_0$ is a free variable to be determined.

Next, for $|i| > 1$, we define:
\begin{equation}
    P_{i,j} = \begin{cases}
    e^\epsilon p &\quad \text{ if } i = j \\
    p  &\quad \text{ otherwise }
    \end{cases}
\end{equation}
where $p$ is a dependent variable determined by $p_0$.

Finally, for $|i|=1$, we define:
\begin{equation}
    P_{i,j} = \begin{cases}
    e^\epsilon p &\quad \text{ if } i = j \\
    p^*  &\quad \text{ if } j = 0 \\
    p  &\quad \text{otherwise}
    \end{cases}
\end{equation}
where $p^*$ is a dependent variable determined by $p_0$. 

From this foundation, we derive all other necessary properties of the mechanism, including $p_i(x)$ and the relationships between the variables $p$, $p^*$, and $p_0$. By substituting these settings into Eq.~\eqref{pij}, $p_{i,j}(x)$ over the interval $x \in [x_{-1}, x_1]$ is defined as follows:
\begin{gather} \label{eq:prob1}
    p_{i, j}(x) =\begin{cases}
        \frac{j(p - e^\epsilon p)+e^\epsilon p - p^*}{x_1}|x| + p^* &\quad \text{if } i=-1 \\[0.4em]    
        \frac{(1-e^\epsilon)p_0}{x_1}|x| + e^\epsilon p_0 &\quad \text{if } i=0 \\[0.4em]
        \frac{j(e^\epsilon p - p) + p - p^*}{x_1}|x| + p^* &\quad \text{if } i=1 \\[0.4em]
        p  &\quad \text{otherwise}
    \end{cases}
\end{gather}
For the remaining intervals $x \in [-1, x_{-1}] \cup [x_1, 1]$, $p_{i,j}(x)$ is defined as follows:
\begin{gather} \label{eq:prob2}
    p_{i,j}(x) = \begin{cases}
        \frac{x - ta_{j-1}}{a_j - a_{j-1}} + p  &\quad \text{if } i = j \\[0.4em]
        \frac{ta_j - x}{a_{j} - a_{j-1}} + p &\quad \text{if } i = j-1\\[0.4em]
        p &\quad \text{otherwise}
    \end{cases}
\end{gather}
where $t = (e^\epsilon - 1)p$ and $x_j = ta_j$. Given the sum of the probability condition, the variable $p$ is:
\begin{equation*}
    p = \frac{1 - p_0}{e^\epsilon +2n - 1}
\end{equation*}
Similarly, the variable $p^*$ is:
\begin{equation*}
    p^* = \frac{1 - 2(n-1)p - e^\epsilon p_0}{2}
\end{equation*}
This simplification not only reduces the number of free variables from $\mathcal{O}(N^2)$  to $\mathcal{O}(N)$, but also mitigates the interdependence between $p_{i,j}(x)$ and $a_i$.

The N-output mechanism guarantees $\epsilon$-LDP by constraining $0 \leq p_0 \leq p$.
\begin{thm}
    The N-output mechanism is $\epsilon$-LDP.
\end{thm}
\begin{proof}
    For $|i| > 1$, every probability $p_i(x)$ falls within the range $[p, e^\epsilon p]$. For any two input values $x, x'$, we have:
    \begin{equation*}
        max_{x, x'}\frac{p_i(x)}{p_i(x')} = \frac{e^\epsilon p}{p} = e^\epsilon
    \end{equation*}
    Similarly, for $i=0$, we have:
    \begin{equation*}
        max_{x, x'}\frac{p_0(x)}{p_0(x')} = \frac{e^\epsilon p_0}{p_0} = e^\epsilon   
    \end{equation*}
    For $|i| = 1$, we demonstrate that $p \leq p^* \leq e^\epsilon p$ to prove the theorem. Since $p^*$ decreases as $p_0$ increases, and given $0 \leq p_0 \leq p$, we need to show two cases: 
    \begin{enumerate}
        \item For $p_0 = 0$, we have $p = \frac{1}{e^\epsilon + 2n -1}$. Thus, the inequality $p^* \leq e^\epsilon p$ is equivalent to:
        \begin{align*}
            \quad & \frac{1 - 2(n-1)p}{2} \leq e^\epsilon p\\
            \iff & \frac{1}{2} - \frac{n-1}{e^\epsilon + 2n - 1} \leq \frac{e^\epsilon}{e^\epsilon + 2n - 1}\\
            \iff & 1 \leq e^\epsilon
        \end{align*}
        Since $\epsilon > 0$, the inequality $p^* \leq e^\epsilon p$ holds.
        \item For $p_0 = p$, we have $p^* = p$. Thus, the inequality $p \leq p^*$ holds trivially.
    \end{enumerate}
\end{proof}

\textbf{Optimization.}
In the numerical solution we optimize the output sets $\Omega$ and the entire probability set $\mathcal{P}$. In this simplified analytical solution, however, our goal is to find the optimal $\Omega$ and a single free variable $p_0$ that minimize the worst-case noise variance. Formally, we aim to solve the following objective function:
\begin{equation*}
    \underset{\Omega, p_0}{\text{Minimize}}  \max_{x \in [-1, 1]} \text{Var}[Y \mid x]
\end{equation*}

We first analyze the structure of the noise variance function $\text{Var}[Y|x]$. Due to the symmetry of the N-output mechanism around $x=0$, it is sufficient to consider the domain $0 \leq x \leq 1$. Let $\text{Var}_j[Y|x]$ be the noise variance over the interval $x_{j-1} \leq x \leq x_j$ for $j \in \{1, \ldots, n\}$. The expression for $\text{Var}_j[Y|x]$ is a piecewise quadratic function:

i) For $j=1$, $\text{Var}_j[Y|x]$ is
\begin{equation}
    - x^2 + \frac{a_1(e^\epsilon p + p - 2p^*)}{(e^\epsilon-1)p}x + 2a_1^2 p^* + 2\sum_{i=2}^{n} a_i^2 p
\end{equation}

ii) For $j \geq 2$, $\text{Var}_j[Y|x]$ is
\begin{equation}
    - x^2 + (a_{j-1} + a_j)x - (e^\epsilon-1)pa_{j-1}a_j + 2\sum_{i=1}^{n}a_i^2p
\end{equation}

The following lemmas provide conditions required to minimize the maximum noise variance.
\begin{lem} \label{lem:max_point}
    For $j \geq 2$, let $x^*_j = \frac{a_{j-1} + a_j}{2}$. Given $\epsilon$ and $N$, $\text{Var}_j[Y|x]$ must attain its maximum value at $x=x^*_j$ to minimize the worst-case noise variance.
    \begin{proof}
        The proof is provided in Appendix \ref{pf:lem:max_point}
    \end{proof}
\end{lem}
\begin{lem} \label{lem:worst-case}
    To minimize the worst-case noise variance, $\text{Var}_n[Y|x^*_n]$ must be the worst-case noise variance.
    \begin{proof}
        The proof is provided in Appendix \ref{pf:lem:worst-case}
    \end{proof}
\end{lem}
These lemmas imply that the optimization problem simplifies to finding a mechanism where the variance $\Var_n[Y|x_n^*]$ is the global maximum. The following theorem identifies the optimal $\Omega$ under this condition.
\begin{thm} \label{thm:type0}
    Given $\epsilon$, $N$, and $p_0$, the sequence $a_j$ that minimizes the worst-case noise variance is defined as follows. 
    Define the sequences $P_i$ and $Q_i$ as follows:
    \begin{gather*}
        T_n = 0, \quad T_{n-1} = 1 \\
        T_i = (4t - 2)T_{i+1} - T_{i+2} \quad \text{for } i \in \{1, 2, \ldots, n-2\} \\[0.5em]
        Q_n = 1, \quad Q_{n-1} = 0 \\
        Q_i = (4t - 2)Q_{i+1} - Q_{i+2}  \quad \text{for } i \in \{1, 2, \ldots, n-2\}
    \end{gather*}
    Then, the sequence $a_j$ is given by:
    \begin{gather*}
        a_n = \frac{1}{t}, \quad a_{n-1} = \frac{(2t - 1) - 8p \sum_{i=1}^n T_i Q_i}{1 + 8p \sum_{i=1}^n T_i^2} a_n \\
        a_j = (4t - 2)a_{j+1} - a_{j+2} \quad \text{for } j \in \{1, 2, \ldots, n-2\}
    \end{gather*}
    This holds only if the following conditions are satisfied:
    \begin{gather}
        a_j < a_{j+1} \quad \text{for } j \in \{0, 1, \ldots, n-1\} \label{type0_c1} \\
        \text{Var}_n[Y \mid x^*_n; p_0] \geq \text{Var}_1[Y \mid x^*_1; p_0] \label{type0_c2}
    \end{gather}
    where $x^*_1 = \frac{a_1(e^\epsilon p + p - 2p^*)}{2(e^\epsilon-1)p}$
    \begin{proof}
        The proof is provided in Appendix \ref{pf:thm:type0}
    \end{proof}
\end{thm}
\noindent The validity of this solution depends on the parameters $N$ and $\epsilon$, and requires a different strategy for even and odd $N$. 

For even $N$, the structure requires $p_0 = 0$. The solution using Theorem~\ref{thm:type0} is optimal if the condition \eqref{type0_c1} and \eqref{type0_c2} are met. If condition \eqref{type0_c1} fails, $N$ it too large. If condition \eqref{type0_c2} fails, the solution is not optimal, where the optimal solution is given by Theorem~\ref{thm:type1}.

For odd $N$, $p_0 \in [0,p]$ is a free variable. To check the validity of the solution, we have three scenarios:
\begin{itemize}
    \item If $\text{Var}_n[Y | x^*_n; p_0=0] > \text{Var}_1[Y | x^*_1; p_0=0]$, this indicates that $N$ is too large relative to the current $\epsilon$.

    \item If $\text{Var}_n[Y | x^*_n; p_0=p] \geq \text{Var}_1[Y | x^*_1; p_0=p]$, then the solution using Theorem~\ref{thm:type0} is optimal.

    \item If $\text{Var}_n[Y | x^*_n; p_0=p] < \text{Var}_1[Y | x^*_1; p_0=p]$, then Theorem \ref{thm:type1} offers the appropriate setting to minimize the worst-case noise variance.
\end{itemize}
The value of $p_0$ that minimizes the worst-case noise variance is given by:
\begin{equation} \label{eq:opt_p0}
    p_0 = \argmin_{p_0 \in [0, p]} \left|\text{Var}_n[Y \mid x^*_n; p_0] - \text{Var}_1[Y \mid x^*_1; p_0]\right|
\end{equation}

When the validity conditions of Theorem~\ref{thm:type0} cannot be met, Theorem~\ref{thm:type1} provides the optimal solution, by setting $p_0=0$ when $N=2n$ and $p_0=p$ when $N=2n+1$.
\begin{thm} \label{thm:type1}
    Given $p=\frac{1}{e^\epsilon +N -1}$, if $\Omega$ derived by Theorem \ref{thm:type0} satisfies $a_j < a_{j+1}$ and $\text{Var}_1[Y|x^*_1; p] > \text{Var}_n[Y|x^*_n; p]$, then the worst-case noise variance is minimized by the following sequences. Define the sequence $C_j$ as:
    \begin{gather*}
         C_1 = \begin{cases}
             \frac{1}{4t-1} & \text{if } N=2n \\
             \frac{1}{4t-2} & \text{if } N=2n+1 \\
        \end{cases}, \\
        C_{j+1} = \frac{1-2t + \sqrt{C_{j}^2 + 2C_{j} - 4tC_{j} + (2t-1)^2}}{C_{j}^2 + 2C_{j} - 4tC_{j}}
    \end{gather*}
    Then, the sequence $a_j$ is defined as:
    \begin{gather*}
        a_n = \frac{1}{t}, \quad
        a_j = C_ja_{j+1} \quad \text{for } j \in \{1, 2, \ldots, n-1\}
    \end{gather*}
    \begin{proof}
        The proof is provided in Appendix \ref{pf:thm:type1}
    \end{proof}
\end{thm}
\textbf{Optimization algorithm.} Based on Theorems \ref{thm:type0} and \ref{thm:type1}, we provide the detail algorithm to find the optimal $N$ and $\Omega$ in Appendix \ref{alg:optimize}.

\textbf{Lower bound of variance.} The structured N-output mechanism also achieves a favorable lower bound on the worst-case noise variance, which improves as $N$ increases. 
\begin{lem} \label{lem:var_infty}
    The infimum of the worst-case noise variance of the N-output mechanis is:
    \begin{equation*}
        \inf\max_{x \in [-1, 1]} \text{Var}[Y|x;N, \Omega] = \frac{1}{(N-1)^2}
    \end{equation*}
\begin{proof}
    The proof is provided in Appendix \ref{pf:lem:var_infty}
\end{proof}
\end{lem}
Following this lemma, prior mechanisms like Duchi's~\cite{duchi2018minimax} and three-output~\cite{zhao2020local} have fixed lower bounds of 1 and $\frac{1}{4}$, respectively, even as $\epsilon \to \infty$. In contrast, the worst-case noise variance of the N-output mechanism can be made small by increasing $N$.

\subsection{Analysis}
\begin{figure}[t!]
    \centering
    \includegraphics[width=0.49\textwidth, trim={0.3cm 0.3cm 0.3cm 0.3cm}, clip]{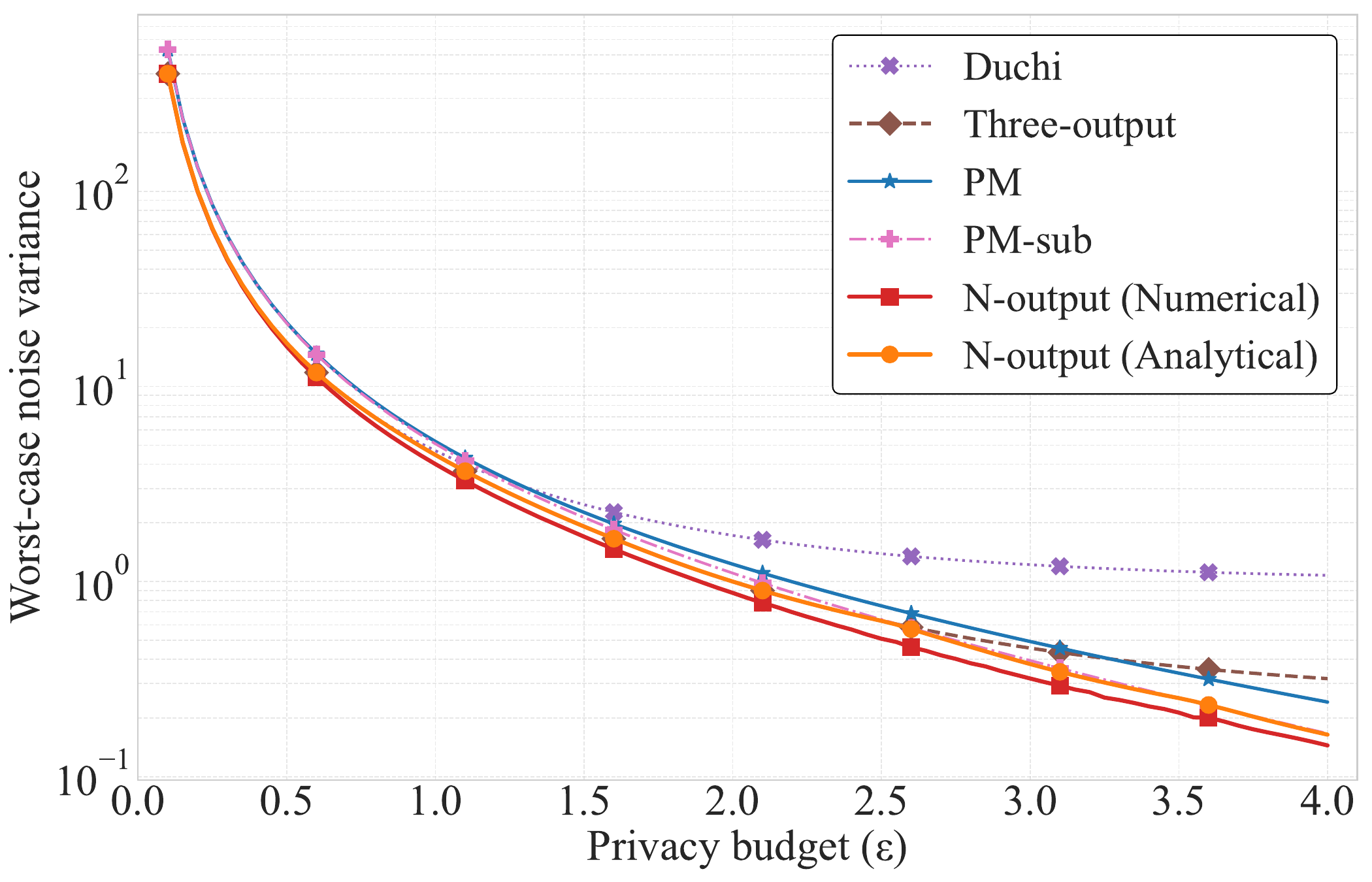}
    \caption{Theoretical worst-case noise variance (Lower is better).}
    \label{fig:wcnv} 
\end{figure}

In this section, we analyze the key properties of our proposed mechanisms, including the output size $N$ and the theoretical worst-case noise variance. 

\textbf{Selecting N.} For the numerical approach, the optimal $N$ for a given $\epsilon$ can be pre-computed by optimizing the mechanism across a range of $N$ values and selecting the one that yields the lowest objective function value. This process is agnostic to datasets and only needs to be performed once per $\epsilon$. Our findings indicate that the N-output mechanism achieves strong performance with $N \leq 16$. While a larger $N$ can offer marginal accuracy improvements, we consider $N \leq 16$ a practical upper bound for communication efficiency; for instance, it requires at most 4 bits to encode each report. Thus, we focus on $N \leq 16$ in this paper.

For the analytical approach, $N$ is similarly chosen to minimize the objective function. However, due to the simplifications made in the analytical model, the optimal $N$ is typically smaller than in the numerical approach. While this us sufficient for mean estimation, it can be a significant limitation for tasks like distribution estimation that benefit from a larger output space. As summarized in Table~\ref{table:bits}, the required number of bits $(\log_2 N)$ for the analytical solution is often limited to $N \in \{2, 3, 4, 5\}$ in practical $\epsilon$ ranges.

\begin{table}[h!]
\centering
\caption{Required bits according to $\epsilon$}
\begin{tabular}{ccccc}
\hline \hline
    & $\epsilon$ < 0.69 & $\epsilon$ < 2.54 & $\epsilon$ < 5.41 & $\epsilon$ < 7.8 \\ \hline
\textbf{bits} & 1     & 2  & 3     & 4  \\ \hline \hline
\end{tabular}
\label{table:bits}
\end{table}

\textbf{Worst-case noise variance.} We compare the theoretical worst-case noise variance of our mechanism with existing ones, as shown in Figure~\ref{fig:wcnv}. Our N-output mechanisms, particularly the numerical version, consistently outperform other mechanisms across all $\epsilon$ values.

We observe distinct behaviors among different types of mechanisms. Continuous output mechanisms (PM~\cite{wang2019collecting} and PM-sub~\cite{zhao2020local}) perform poorly at low $\epsilon$ but improves as $\epsilon$ increases. Conversely, other discrete mechanisms (Duchi's~\cite{duchi2018minimax} and three-output~\cite{zhao2020local}) are effective at low $\epsilon$ but show limited improvement at high $\epsilon$, where their variance fails to decrease significantly. In contrast, the variance of our N-output mechanism decreases steadily as $\epsilon$ grows, a result of its ability to adaptively increase $N$. This empirical result aligns with our theoretical findings in Lemma~\ref{lem:var_infty}. Consequently, the N-output mechanism not only achieves the lowest theoretical worst-case noise variance in the practical $\epsilon$ range but also does so while being more communication-efficient than continuous-output mechanisms.

\section{Advancing Distribution Estimation} \label{sec:em}
In this section, we extend the scope of our analysis from mean estimation to the distribution estimation. mechanisms such as Duchi's~\cite{duchi2018minimax}, three-output~\cite{zhao2020local}, and PMs~\cite{wang2017locally, zhao2020local} were primarily designed for mean estimation, their collected data can be post-processed by the server to estimate the underlying data distribution. In contrast, the Square Wave (SW) mechanism \cite{li2020estimating} was specifically designed for distribution estimation, employing the Expectation-Maximization (EM) algorithm \cite{dempster1977maximum} to derive the distribution from perturbed data. The goal of SW mechanism is to minimize the Wasserstein distance.

However, a generalized application of EM to other existing mechanisms has not been fully addressed. Here, we build upon the process from \cite{li2020estimating} to propose a unified EM-based framework for distribution estimation. This framework is generalized to address two primary scenarios, applying to both continuous-output mechanisms and discrete-output mechanisms.

\textbf{Problem definition.} Consider an aggregator receiving $m$ perturbed values from users. The users' true values are drawn from a domain $\mathcal{X}= [-1, 1]$.  To model the data distribution, we partition this domain into $d$ disjoint bins. The true distribution is represented by a vector $\boldsymbol{\pi} = \{\pi_1 , \ldots, \pi_d\}$, where each $\pi_i$ is the proportion of true values that fall into the $i$-th bin. The objective of distribution estimation is to accurately compute an estimate of this vector, denoted as  $\boldsymbol{\pi} = \{\hat{\pi}_1, \ldots, \hat{\pi}_d\}$, which maximizes the log-likelihood $\mathcal{L(\hat{\pi})} = \ln Pr[\mathbf{y} | \hat{\pi}]$.

\textbf{EM framework.} The input domain and output space of continuous-output mechanism are divided into $d$ bins and $\Tilde{d}$ bins, respectively. Let $R_i= [r_{i-1}, r_i]$ denote the $i$-th bin of the divided input domain and $\Tilde{R}_i= [\Tilde{r}_{i-1}, \Tilde{r}_i]$ denote the $i$-th bin of the divided output space. Let $c_j$ represent the number of outputs falling into $\Tilde{R}_i$. 
For discrete-output mechanisms, the $N$ possible outputs naturally form $N$ bins, where $c_j$ is the count of the discrete output $a_i$.
The algorithm for distribution estimation is shown in Algorithm~\ref{alg:em}

\begin{algorithm}
\caption{Distribution estimation employing EM} \label{alg:em}
    \begin{algorithmic}[1]
    \Require $y, \Theta^\mathcal{P}, \Theta^\mathcal{N}, \tau,$
    \Ensure $\hat{\pi}$

    \While{$|L_{v+1}(\hat{\pi}) - L_{v}(\hat{\pi})| > \tau$}
        \For{$i \in \{1, \ldots, d\}$}
            \If{Continuous mechanism}
                \State $G_i = \hat{\pi}_i \sum_{j=1}^{\Tilde{d}} c'_j \frac{\Theta_{j, i}^{\mathcal{P}}}{\Sigma_{k=1}^d \Theta_{j, k}^{\mathcal{P}} \mu_k}$\
            \ElsIf{Discrete mechanism}
                \State $G_i = \hat{\pi}_i \Sigma_{j=-n}^{n} c_j \frac{\Theta_{j, i}^\mathcal{N}}{\Sigma_{k=1}^d \Theta_{j, k}^\mathcal{N} \hat{\pi}_k}$\
            \EndIf
            
        \EndFor
        \For{$i \in \{1, \ldots, d\}$}
            \State $\hat{\pi}_i = \frac{G_i}{\Sigma_{k=1}^d G_k}$
        \EndFor
    \EndWhile
    \State \Return $\hat{\pi}$
    \end{algorithmic}
\end{algorithm}

Following the EM process outlined \cite{li2020estimating}, we define $Q_i$ for continuous-output mechanisms as:
\begin{equation*}
    G_i = \mu_i \sum_{j=1}^{\Tilde{d}} c'_j \frac{\text{Pr}[y \in \Tilde{R}_j | x \in R_i, \mu]}{\text{Pr}[y \in \Tilde{R}_j |\mu]}
\end{equation*}
where $\Theta^\mathcal{P}_{j, i} = \text{Pr}[y \in \Tilde{R}_j | x \in R_i, \mu]$. 
For discrete-output mechanism, we define $Q_i$ as:
\begin{equation*}
    G_i = \hat{\pi}_i \sum_{j=-n}^{n} c_j \frac{\text{Pr}[y = a_j| x \in R_i, \mu]}{\text{Pr}[y = a_j |\mu]}
\end{equation*}
where $\Theta^\mathcal{N}_{j, i} = \text{Pr}[y a_j | x \in R_i, \mu]$ and $c_j$ is the number of outputs $a_j$. Here, $\Theta^{\mathcal{P}}$, $\Theta^{\mathcal{N}}$ are the probability transition matrices for the continuous and discrete cases, respectively.

The distribution $\hat{\pi}$ is then updated iteratively as:
\begin{equation*}
    \hat{\pi}_i = \frac{G_i}{\Sigma_{k=1}^d G_k}
\end{equation*}
The iteration stops when the relative improvement in the log-likelihood becomes sufficiently small, as suggested by \cite{fanti2016building, li2020estimating}

\textbf{Derivation of $\Theta^{\mathcal{P}}$.} Let $l = r_i - r_{i-1}$ and $\Tilde{l} = \Tilde{r}_i - \Tilde{r}_{i-1}$. Here, we derive a general probability transition matrix for the continuous case $\Theta_{i, j}^{\mathcal{P}}$:
\begin{equation*}
\begin{split}
    \Theta_{i, j}^{\mathcal{P}} &= \text{Pr}[y \in \Tilde{R}_j | x \in R_i, \boldsymbol{\pi}] \\
    &= \frac{\text{Pr}[y \in \Tilde{R}_j, x \in R_i  | \boldsymbol{\pi}]}{\text{Pr}[x \in R_i| \boldsymbol{\pi}]}\\
    &= \frac{\int_{R_i} \text{Pr}[y \in \Tilde{R}_j | x, \boldsymbol{\pi}] f(x|\boldsymbol{\pi}) \ dx}{\pi_i}
\end{split}
\end{equation*}
where $f(x|\boldsymbol{\pi})$ is the probability density function. We assume that true values in the range $R_i$ are uniformly distributed, so $f(x|\boldsymbol{\pi}) = \frac{\pi_i}{l}$. Then,
\begin{equation*}
    \Theta_{i, j}^{\mathcal{P}} = \frac{\int_{R_i} \Pr[y \in \Tilde{R}_j | x, \boldsymbol{\pi}] \frac{\pi_i}{l} \ dx}{\pi_i} = \frac{1}{l} \int_{R_i} \Pr[y \in \Tilde{R}_j | x, \boldsymbol{\pi}] \ dx
\end{equation*}
Since $\text{Pr}[y \in \Tilde{R}_j | x, \boldsymbol{\pi}]$ is independent of $\boldsymbol{\pi}$,
\begin{equation} \label{app:eq:theta_int}
    \begin{split}
        \Theta_{i, j}^{\mathcal{P}}  &= \frac{1}{l} \int_{R_i} \text{Pr}[y \in \Tilde{R}_j | x] \ dx \\
        & = \frac{1}{l} \int_{R_i} \int_{\Tilde{R}_j} \text{Pr}[y | x] \ dy \ dx \\
        & = \frac{1}{l} \int_{R_i} \int_{\Tilde{R}_j} q \ dy \ dx + \frac{1}{l} \int_{R_i} \int_{\Tilde{R}_j} \text{Pr}[y|x] - q \ dy \ dx
    \end{split}
\end{equation}
where $q = \min_x \Pr[y|x]$
The first term of \eqref{app:eq:theta_int} is:
\begin{equation} \label{app:eq:theta_first_term}
\begin{split}
    \frac{1}{l} \int_{r_{i-1}}^{r_i} \int_{\Tilde{r}_{j-1}}^{\Tilde{r}_j} q \ dy \ dx
     &= \frac{1}{l} \int_{r_{i-1}}^{r_i} q(\Tilde{r}_{j} - \Tilde{r}_{j-1}) \ dx \\
     &= \frac{1}{l} q(\Tilde{r}_{j} - \Tilde{r}_{j-1})(r_{i} - r_{i-1}) \\
     &= q \Tilde{l}
\end{split}
\end{equation}
The second term of \eqref{app:eq:theta_int} depends on $R_i$ and $\Tilde{R}_j$. To represent this term, we denote the inverse function of $L(\cdot)$ and $R(\cdot)$, which are defined in Appendix~\ref{sec:add_prelim}, as $\Bar{L}(\cdot)$ and $\Bar{R}(\cdot)$, respectively, and use $\mathbb{I}_{(\cdot)}$ as an indicator function. There are two cases according to the value of $\epsilon$.  Let $s=(e^\epsilon-1)$

i) If $R(x) - L(x) \geq \Tilde{l}$,
\begin{equation*}
\begin{split}
    & \quad \int_{R_i} \int_{\Tilde{R}_j} \text{Pr}[y|x] - q \ dy \ dx \\
    & = \mathbb{I}_{ R(r_{i-1}) \leq \Tilde{r}_{j} \bigwedge R(r_i) \geq \Tilde{r}_{j-1} } \int_{ \max \{\Bar{R}(\Tilde{r}_{j-1}) , r_{i-1}\}  }^{ \min\{ \Bar{R}(\Tilde{r}_{j}) , r_i \} } (R(x) - \Tilde{r}_{j-1}) sq  \ dx \\[0.5em]
    & + \mathbb{I}_{ L(r_{i-1}) \leq \Tilde{r}_{j-1} \bigwedge R(r_i) \geq \Tilde{r}_j } \int_{ \max \{ \Bar{R}(\Tilde{r}_j), r_{i-1} \} } ^ { \min \{  \Bar{L}(\Tilde{r}_{j-1}), r_i \} } \Tilde{l}sq \ dx \\[0.5em]
    & + \mathbb{I}_{ L(r_{i-1}) \leq \Tilde{r}_j \bigwedge L(r_i) \geq \Tilde{r}_{j-1}} \int_{ \max \{ \Bar{L}(\Tilde{r}_{j-1}) , r_{i-1} \} } ^ { \min \{ \Bar{L}(\Tilde{r_j}) , r_i \}  }  ( \Tilde{r}_j - L(x) )sq  \ dx
\end{split}
\end{equation*}
ii) If $R(x) - L(x) < \Tilde{l}$,
\begin{equation*}
\begin{split}
    & \quad \int_{R_i} \int_{\Tilde{R}_j} \text{Pr}[y|x] - q \ dy \ dx \\
    & = \mathbb{I}_{ L(r_{i-1}) \leq \Tilde{r}_{j-1} \bigwedge R(r_i) \geq \Tilde{r}_{j-1} } \int_{ \max \{\Bar{R}(\Tilde{r}_{j-1}) , r_{i-1}\}  }^{ \min\{ \Bar{L}(\Tilde{r}_{j-1}) , r_i \} } (R(x) - \Tilde{r}_{j-1})sq  \ dx \\[0.5em]
    & + \mathbb{I}_{ R(r_{i-1}) \leq \Tilde{r}_{j} \bigwedge L(r_i) \geq \Tilde{r}_{j-1} } \int_{ \max \{ \Bar{L}(\Tilde{r}_{j-1}), r_{i-1} \} } ^ { \min \{  \Bar{R}(\Tilde{r}_{j}), r_i \} } (R(x)-L(x))sq \ dx \\[0.5em]
    & + \mathbb{I}_{ L(r_{i-1}) \leq \Tilde{r}_j \bigwedge R(r_i) \geq \Tilde{r}_{j}} \int_{ \max \{ \Bar{R}(\Tilde{r}_{j}) , r_{i-1} \} } ^ { \min \{ \Bar{L}(\Tilde{r_j}) , r_i \}  }  ( \Tilde{r}_j - L(x) )sq  \ dx
\end{split}
\end{equation*}

\begin{figure}[ht!]
    \centering
    \includegraphics[width=0.49\textwidth, trim={0.3cm 0.0cm 3.0cm 2.0cm}, clip]{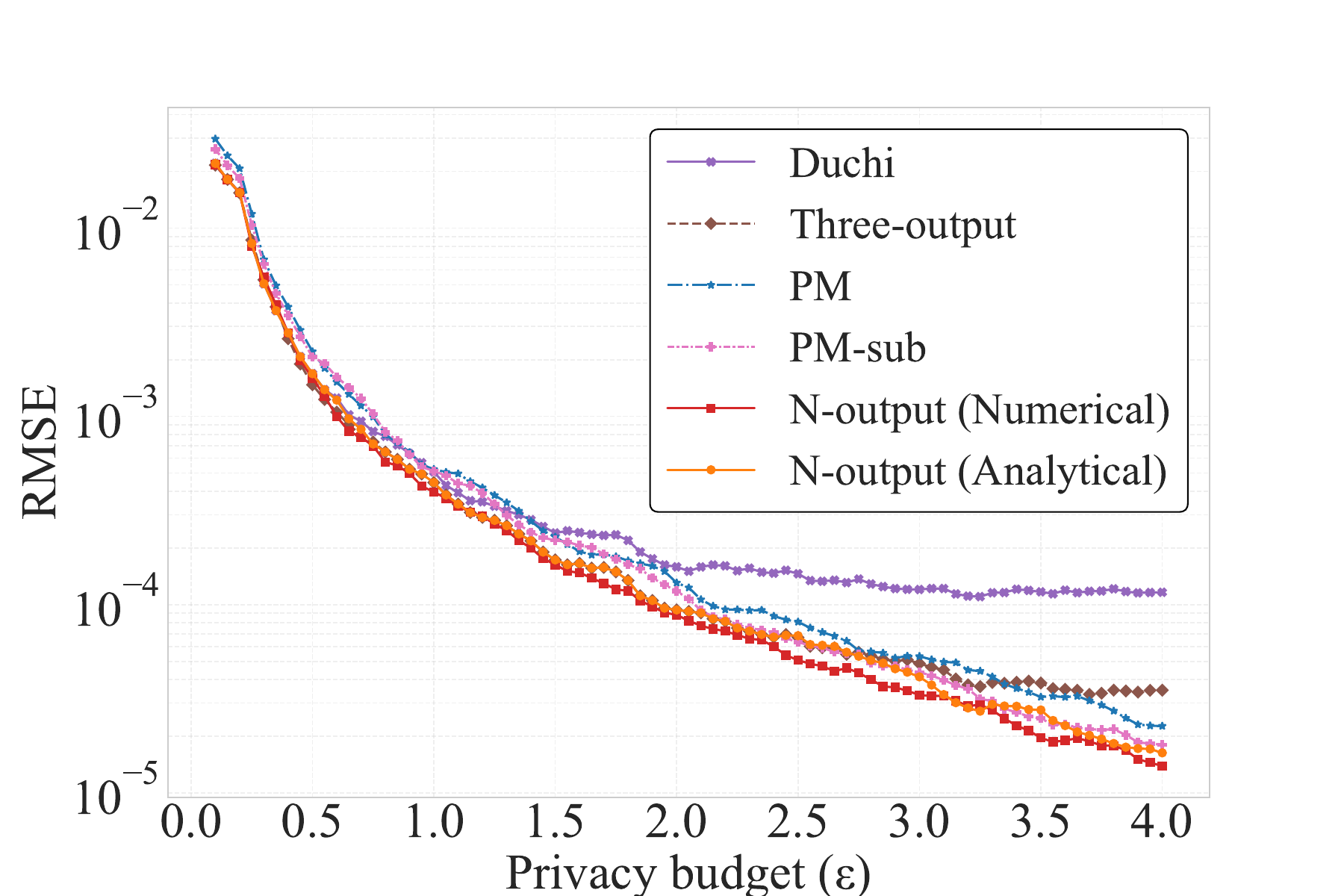}
    \caption{Worst-case mean estimation RMSE. Lower is better.}
    \label{fig:wcrmse} 
\end{figure}

\textbf{Derivation of $\Theta^{\mathcal{N}}$.} For the numerical approach, the key is to define a set of integration points, $Z_i$, for each bin $R_i$. This set is the union of the bin's endpoints $\{r_{i-1}, r_i\}$ and any of the mechanism's endpoints $x_j$ that fall strictly within that bin. Let these sorted points be denoted as $Z_i = \{z_0, z_1, \ldots, z_u\}$.

With this definition, the transition matrix element $\Theta_{ij}^{\mathcal{N}}$ can be expressed:
\begin{equation*}
    \Theta_{ij}^{\mathcal{N}} = \frac{1}{l_i} \sum_{k=0}^{u-1} \left( \frac{p_j(z_k) + p_j(z_{k+1})}{2} \right) (z_{k+1} - z_k)
\end{equation*}
This unified expression seamlessly adapts to all cases. If bin $R_i$ contains no grid points, the set $Z_i$ simplifies to just the two endpoints $\{r_{i-1}, r_i\}$, and the formula automatically reduces to the average probability across the bin:
\begin{align*}
    \Theta_{ij}^{\mathcal{N}} &= \frac{1}{r_i - r_{i-1}} \left( \frac{p_j(r_{i-1}) + p_j(z_i)}{2} \right) (r_{i} - r_{i-1}) \\
    &= \frac{ p_j(r_{i-1}) + p_j(z_i) }{2}
\end{align*}

The probability transition matrix for the analytical approach is in Appendix~\ref{app:em_anal} since analytical approach is not suitable for distribution estimation.

\begin{figure*}[ht!]
    \centering
    \includegraphics[width=\textwidth]{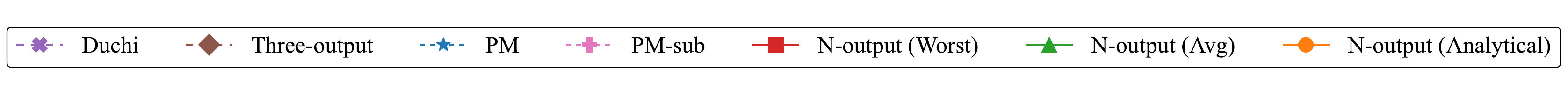}
    \begin{subfigure}[b]{0.497\textwidth}
        \centering
        \includegraphics[width=\textwidth, trim={1.0cm 1.0cm 1.0cm 1.0cm}, clip]{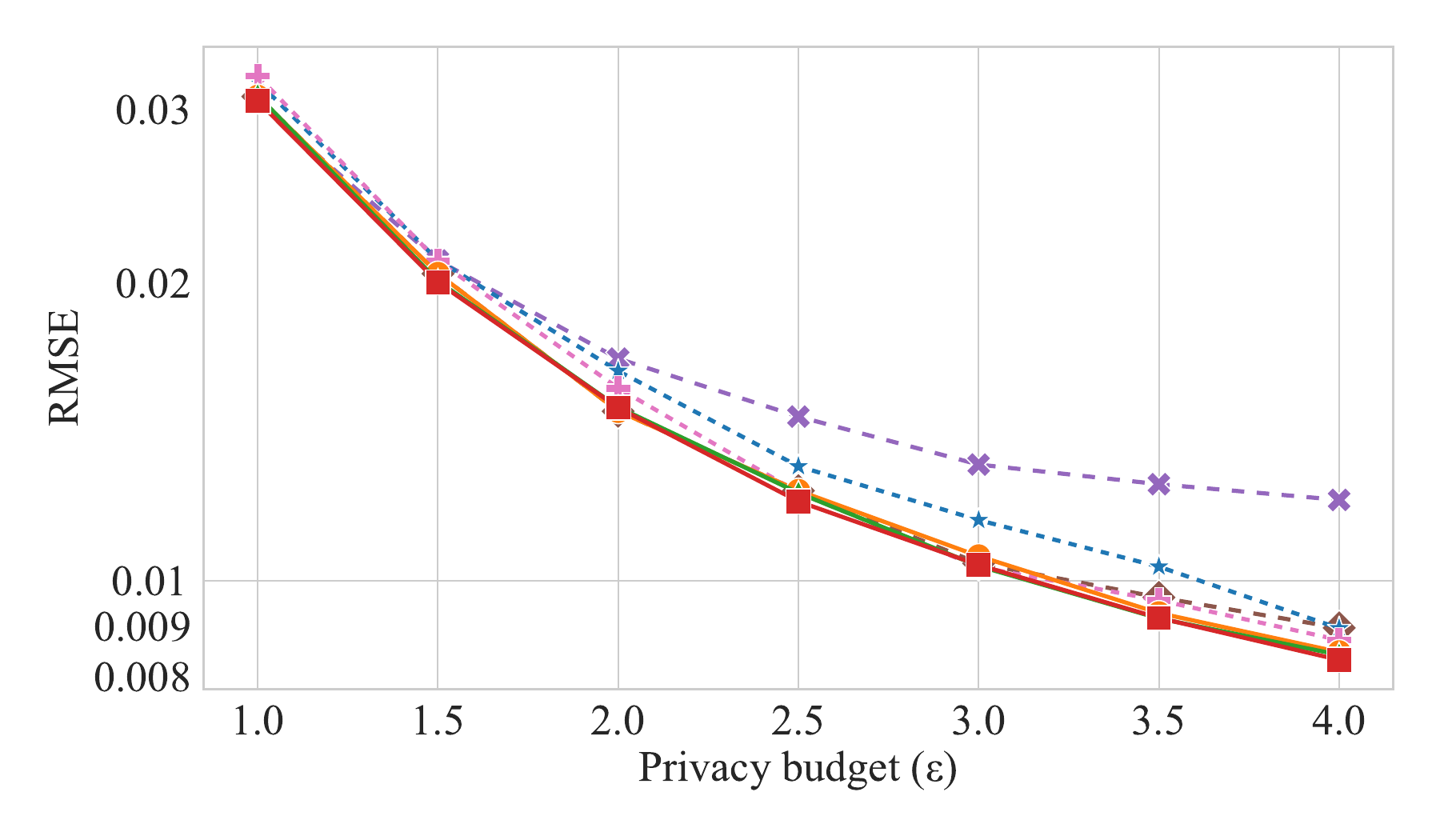} 
        \caption{Adult}
        \label{fig:adult_mean}
    \end{subfigure}
    \begin{subfigure}[b]{0.497\textwidth}
        \centering
        \includegraphics[width=\textwidth, trim={1.0cm 1.0cm 1.0cm 1.0cm}, clip]{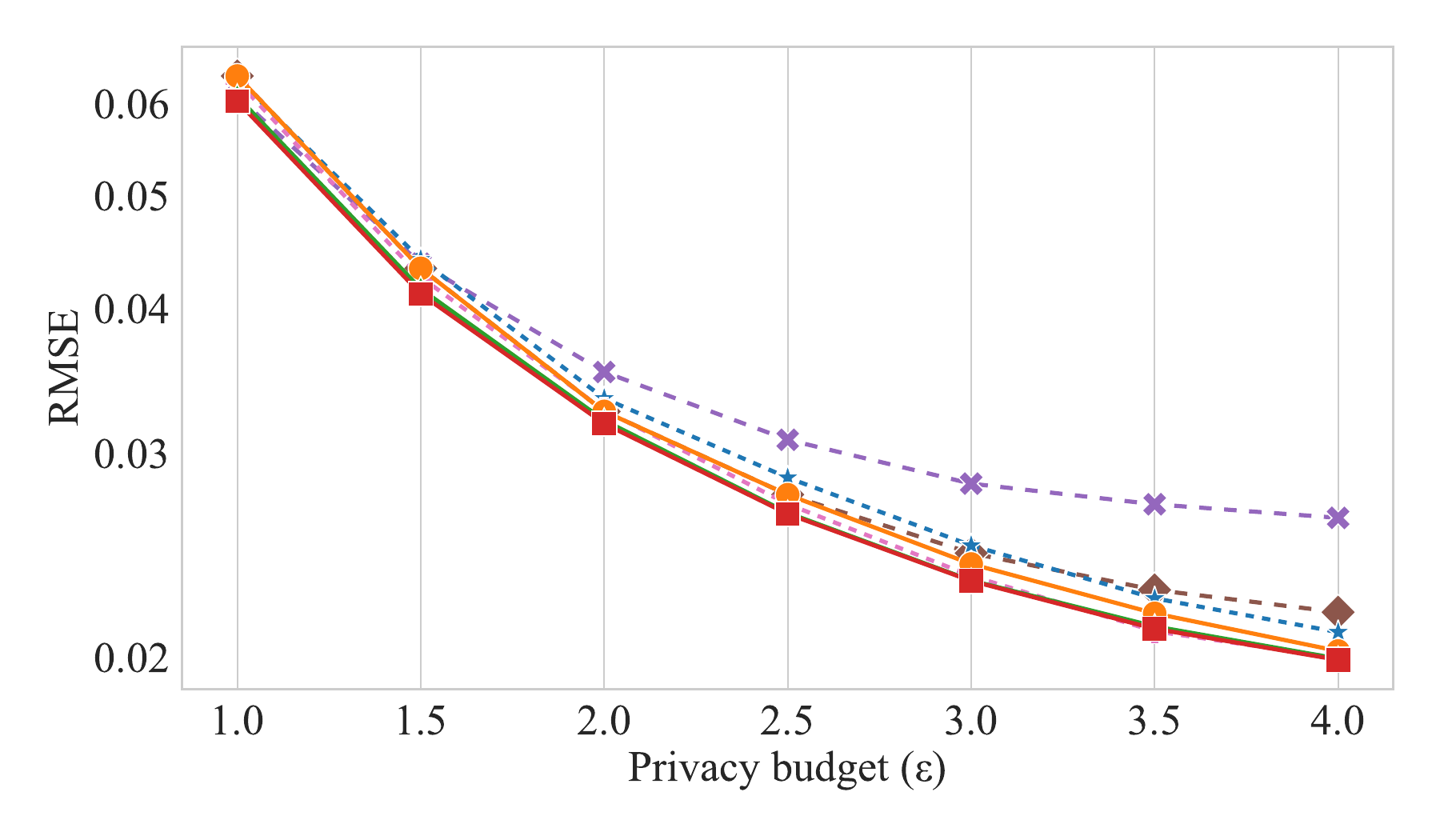} 
        \caption{Credit}
        \label{fig:credit_mean}
    \end{subfigure}
    \begin{subfigure}[b]{0.497\textwidth}
        \centering
        \includegraphics[width=\textwidth, trim={1.0cm 1.0cm 1.0cm 1.0cm}, clip]{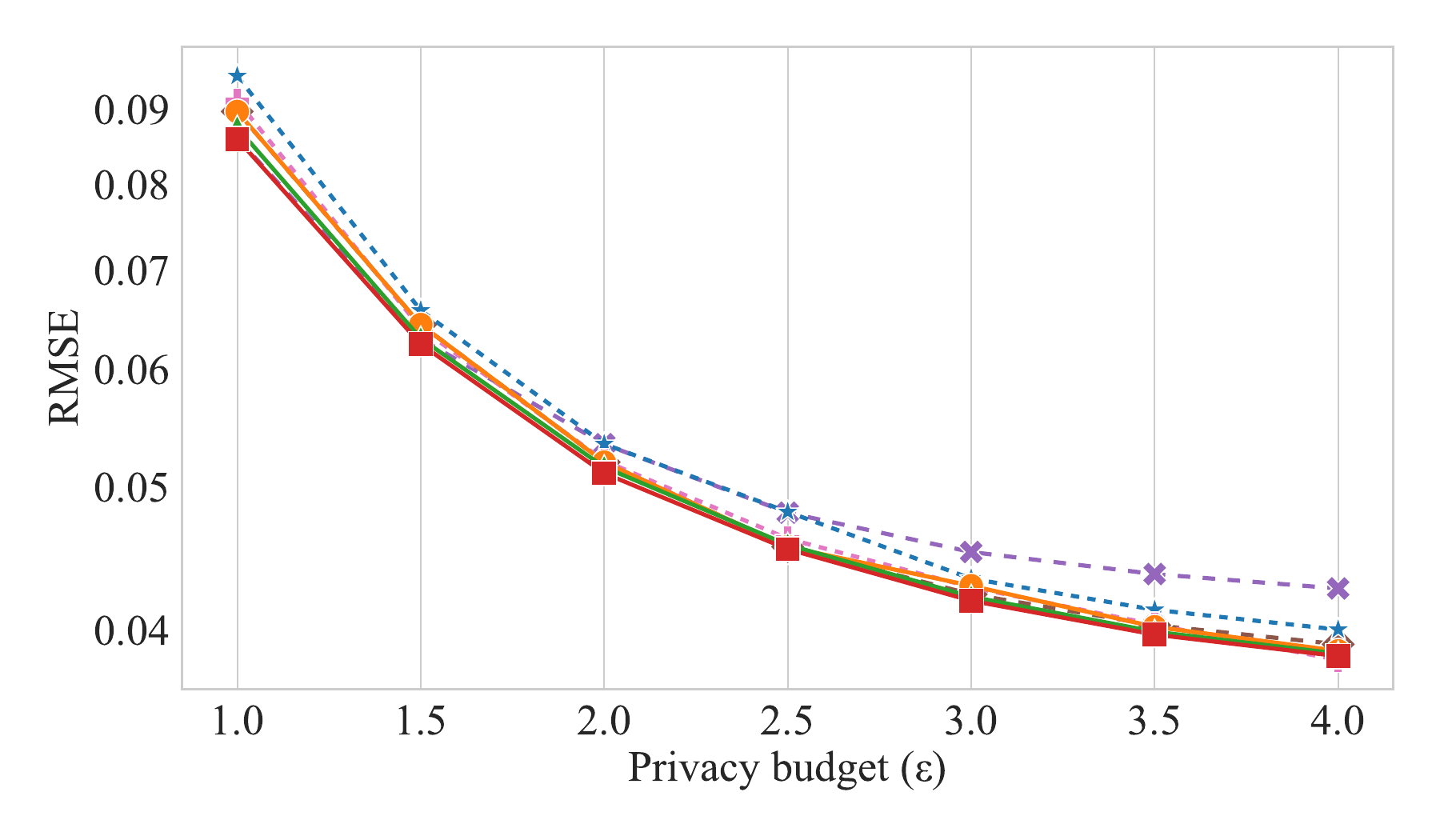} 
        \caption{Online News Popularity}
        \label{fig:ONP_mean}
    \end{subfigure}
    \begin{subfigure}[b]{0.497\textwidth}
        \centering
        \includegraphics[width=\textwidth, trim={1.0cm 1.0cm 1.0cm 1.0cm}, clip]{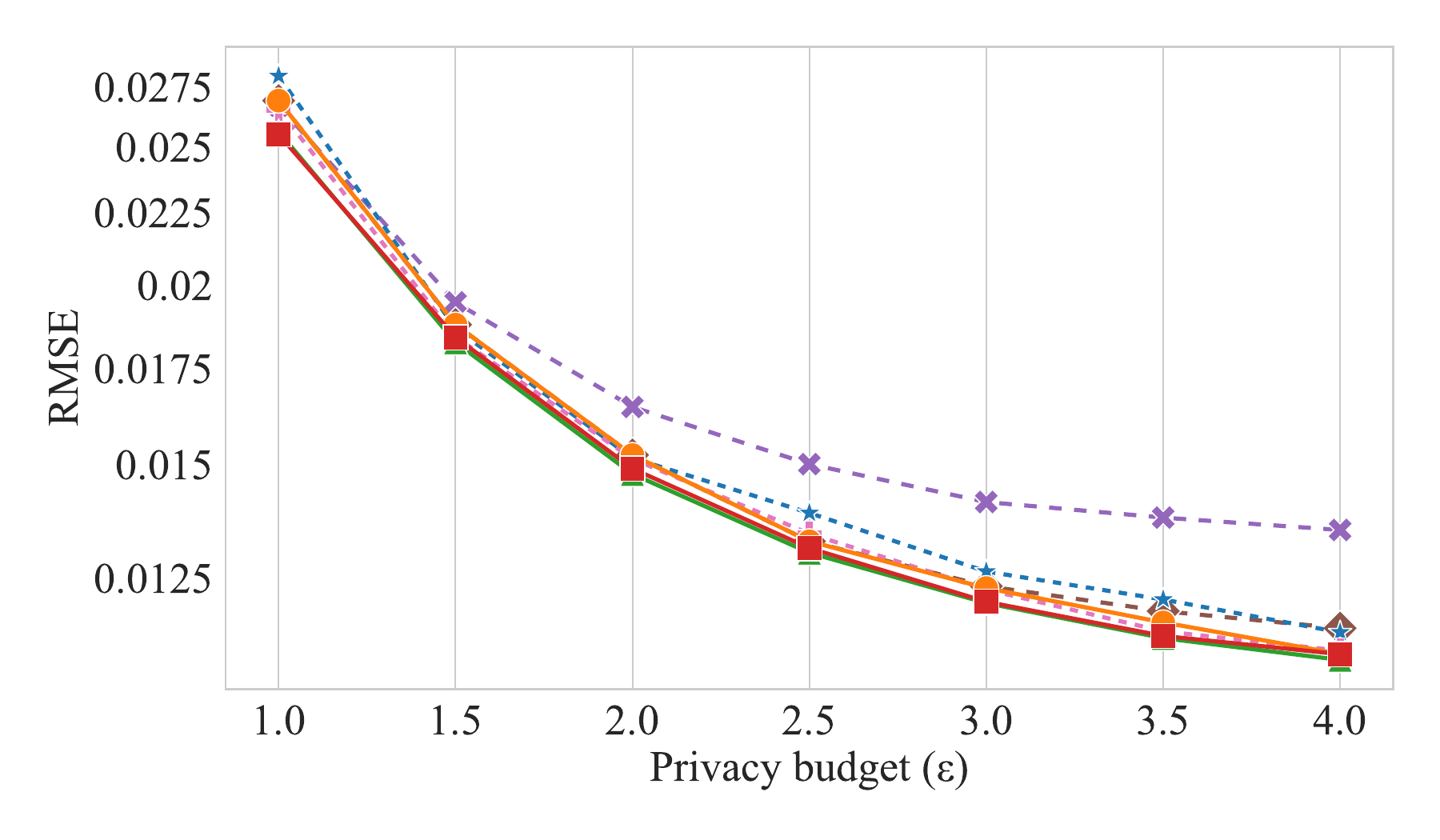} 
        \caption{Beijing Housing}
        \label{fig:BHP_mean}
    \end{subfigure}
    \caption{Mean estimation (RMSE, $\log$ scale) on real-world datasets. Lower is better.}
    \label{fig:exp_mean}
\end{figure*}

\section{Empirical Evaluation}
In this section, we empirically validate the performance of the N-output mechanism across various statistical estimation tasks. We compare our proposed mechanisms with SOTA LDP mechanisms for numerical data. 

\textbf{Experimental setup.} We use four real-world datasets: Adult~\cite{adult_2}, Credit~\cite{default_of_credit_card_clients_350}, Online News Popularity~\cite{online_news_popularity_332}, and Beijing Housing~\cite{ruiqurm_2018}. Since these datasets are multi-dimensional, we adopt a typical sampling method used in prior work~\cite{wang2017locally, zhao2020local}: each client randomly selects a single dimension from their data vector to report. All input values are normalized to the range $[-1,1]$. The results presented are averaged over 100 independent runs. For our mechanisms, the number of outputs $N$ for each $\epsilon$ is selected as the one yielding the lowest objective function value for $N \leq 16$. For reproducibility, our implementation is available in the Open Science section.

\subsection{Mean estimation}
We first validate the theoretical worst-case noise variance (Figure~\ref{fig:wcnv}) through empirical simulation. As shown in Figure~\ref{fig:wcrmse}, this experiment directly measures the RMSE of worst-case mean estimation. The results align perfectly with our analysis: the N-output (Numerical) mechanism consistently achieves the lowest error. In contrast, Duchi's and the Three-output mechanisms show limited accuracy improvements as $\epsilon$ increases, empirically confirming their theoretical limitations. The N-output (Analytical) mechanism also demonstrates strong performance, outperforming PM-sub in the high privacy regime and achieving similar error in the low privacy regime.

Next, we evaluate the performance on four real-world datasets, as shown in Figure~\ref{fig:exp_mean} (plotted on a $\log$ scale). In this evaluation, N-output (Worst) refers to the numerical solution optimized for worst-case variance, while N-output (Avg) is the numerical approach optimized for the average case. The N-output (Worst) mechanism consistently outperforms all other mechanisms across the entire range of $\epsilon$, with the N-output (Avg) mechanism demonstrating nearly identical performance. Additionally, a key advantage of our mechanism is its communication efficiency. While PMs require 64 bits per report in our experiments, the N-output mechanism achieves superior accuracy using at most 4 bits, making it 16 times more bit-efficient. Notably, whereas prior mechanisms like Duchi's, Three-output, and PMs each exhibit strengths in specific privacy regimes, the N-output mechanism demonstrates robust and superior performance across all tested settings. It is worth noting that while the RMSE differences may appear small on the scale, they correspond to significant errors when scaled back to the original data domains.

\begin{figure*}[ht!]
    \centering
    \includegraphics[width=\textwidth]{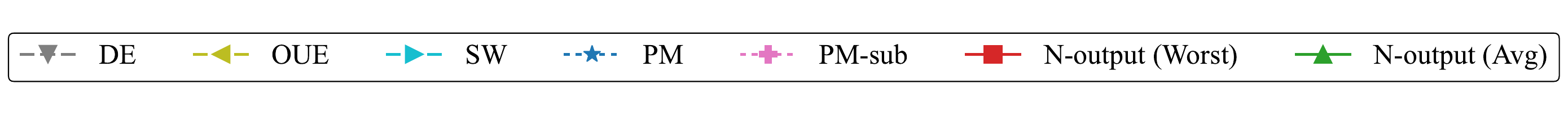}
    \begin{subfigure}[b]{0.497\textwidth}
        \centering
        \includegraphics[width=\textwidth, trim={1.0cm 1.0cm 1.0cm 1.0cm}, clip]{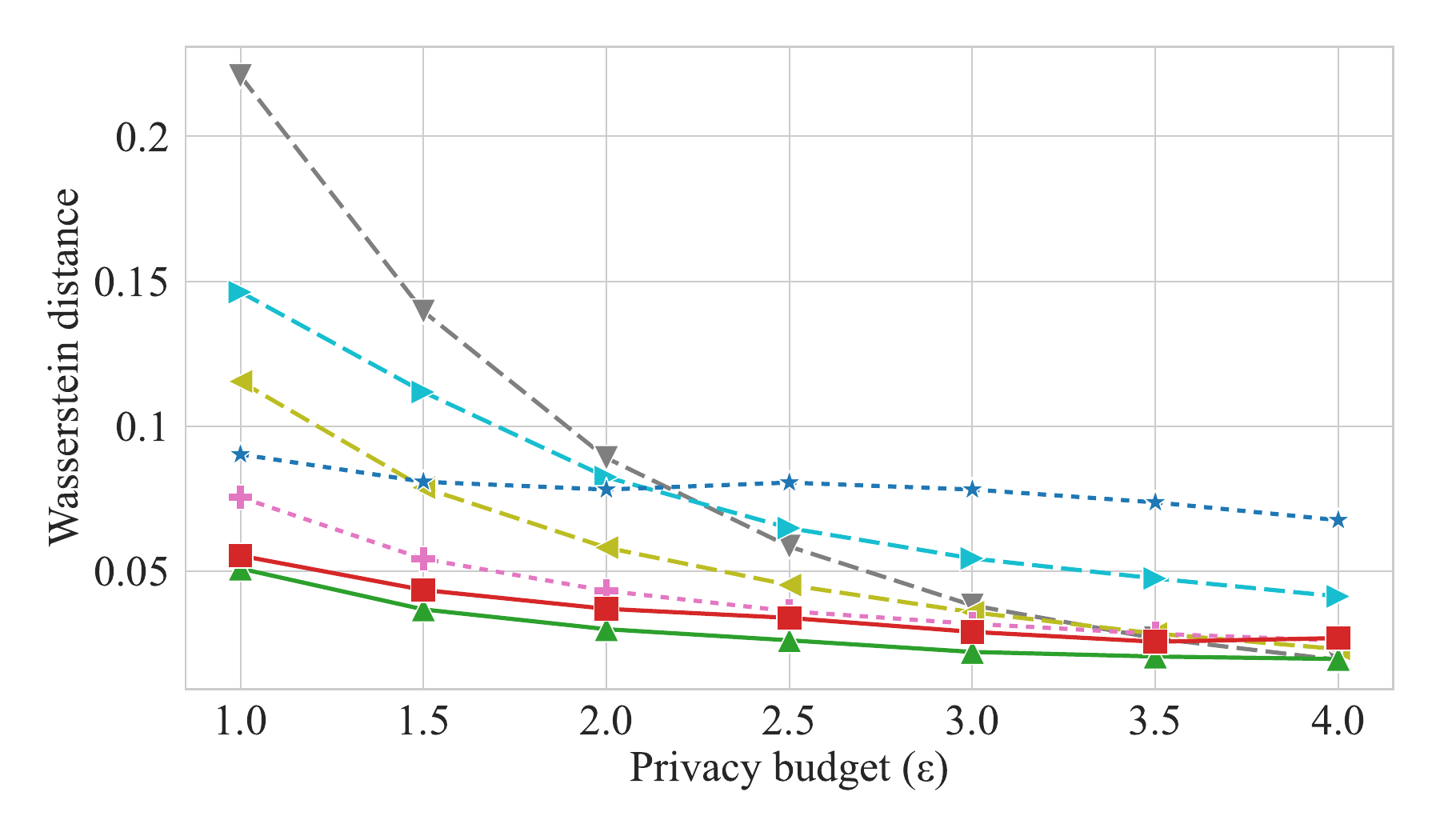} 
        \caption{Adult}
        \label{fig:adult_wass}
    \end{subfigure}
    \begin{subfigure}[b]{0.497\textwidth}
        \centering
        \includegraphics[width=\textwidth, trim={1.0cm 1.0cm 1.0cm 1.0cm}, clip]{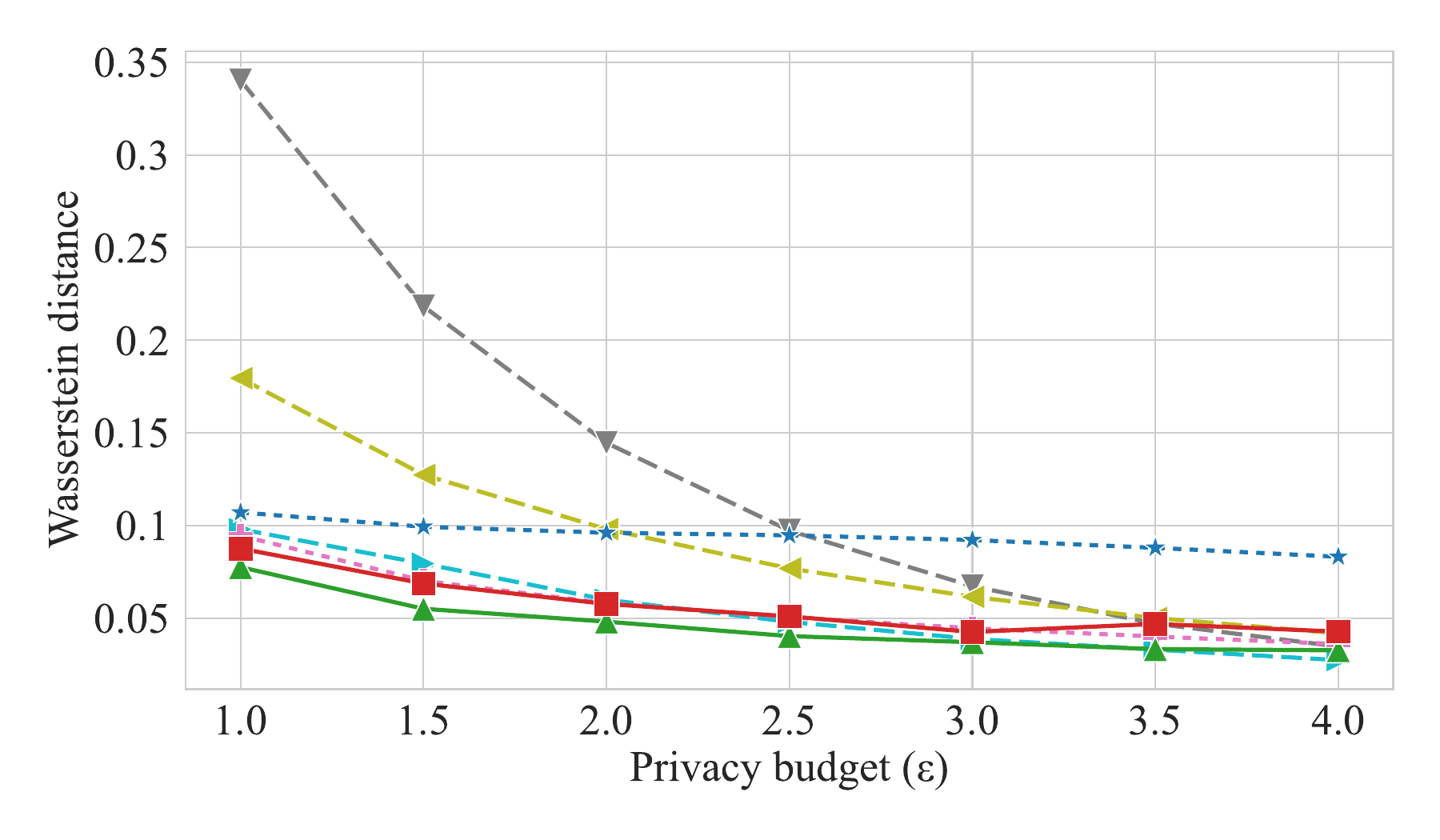} 
        \caption{Credit}
        \label{fig:credit_wass}
    \end{subfigure}
    \begin{subfigure}[b]{0.497\textwidth}
        \centering
        \includegraphics[width=\textwidth, trim={1.0cm 1.0cm 1.0cm 1.0cm}, clip]{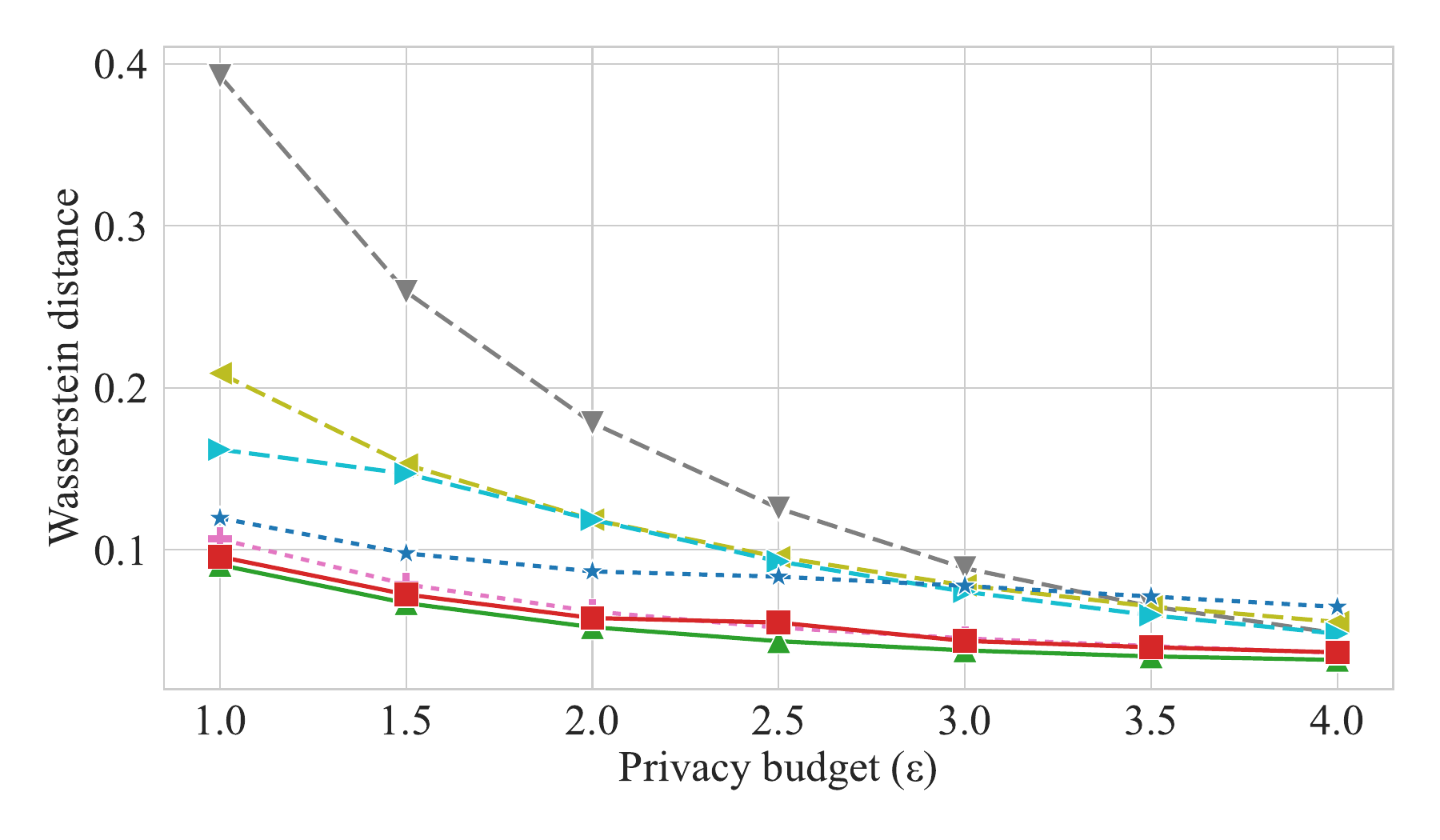} 
        \caption{Online News Popularity}
        \label{fig:ONP_wass}
    \end{subfigure}
    \begin{subfigure}[b]{0.497\textwidth}
        \centering
        \includegraphics[width=\textwidth, trim={1.0cm 1.0cm 1.0cm 1.0cm}, clip]{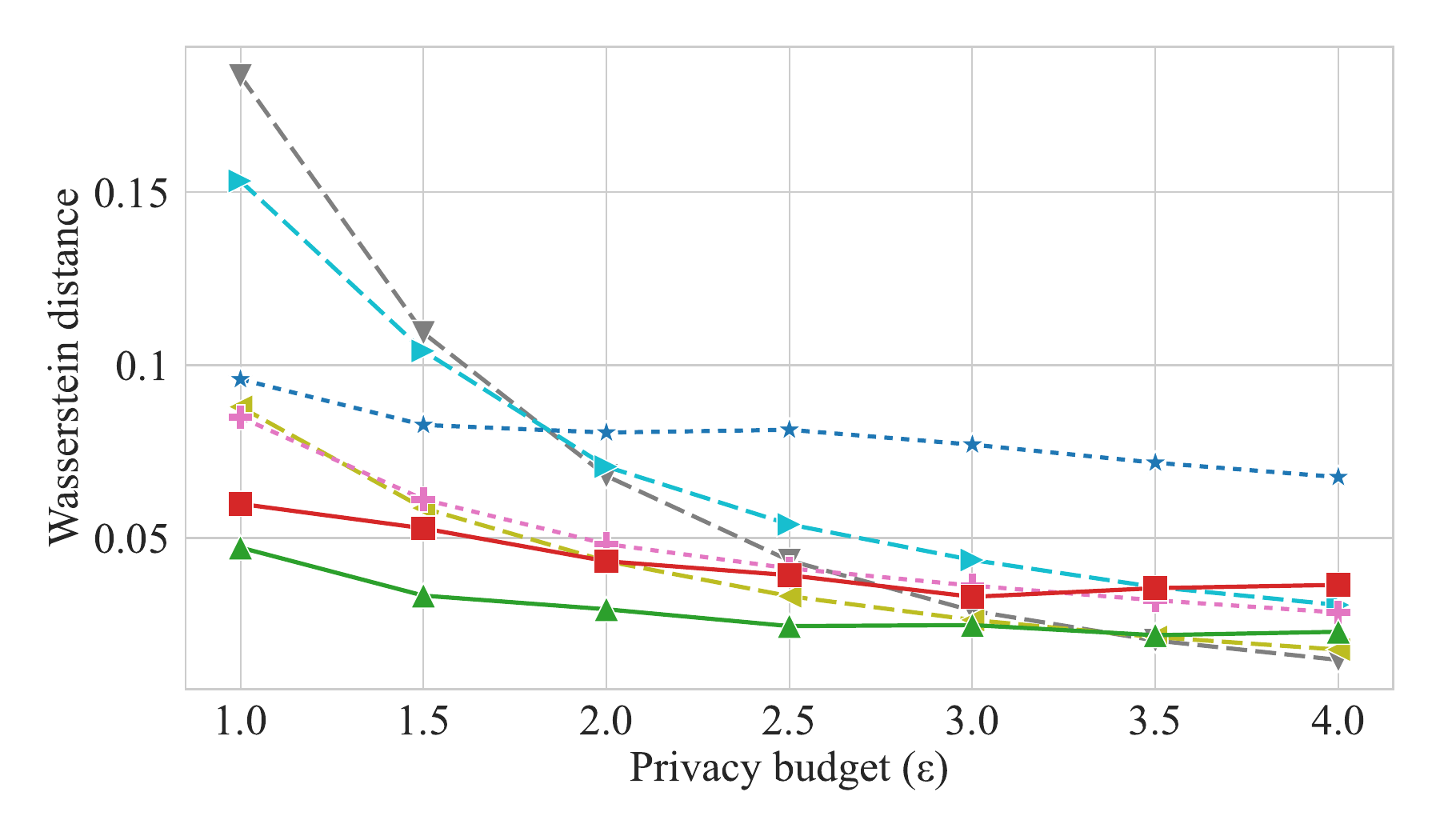} 
        \caption{Beijing Housing}
        \label{fig:BHP_wass}
    \end{subfigure}
    \caption{Distribution estimation (Wasserstein distance) on real-world datasets. Lower is better.}
    \label{fig:exp_wass}
\end{figure*}

\begin{figure*}[ht!]
    \centering
    \includegraphics[width=\textwidth]{figures/legend_dist.pdf}
    \begin{subfigure}[b]{0.497\textwidth}
        \centering
        \includegraphics[width=\textwidth, trim={1.0cm 1.0cm 1.0cm 1.0cm}, clip]{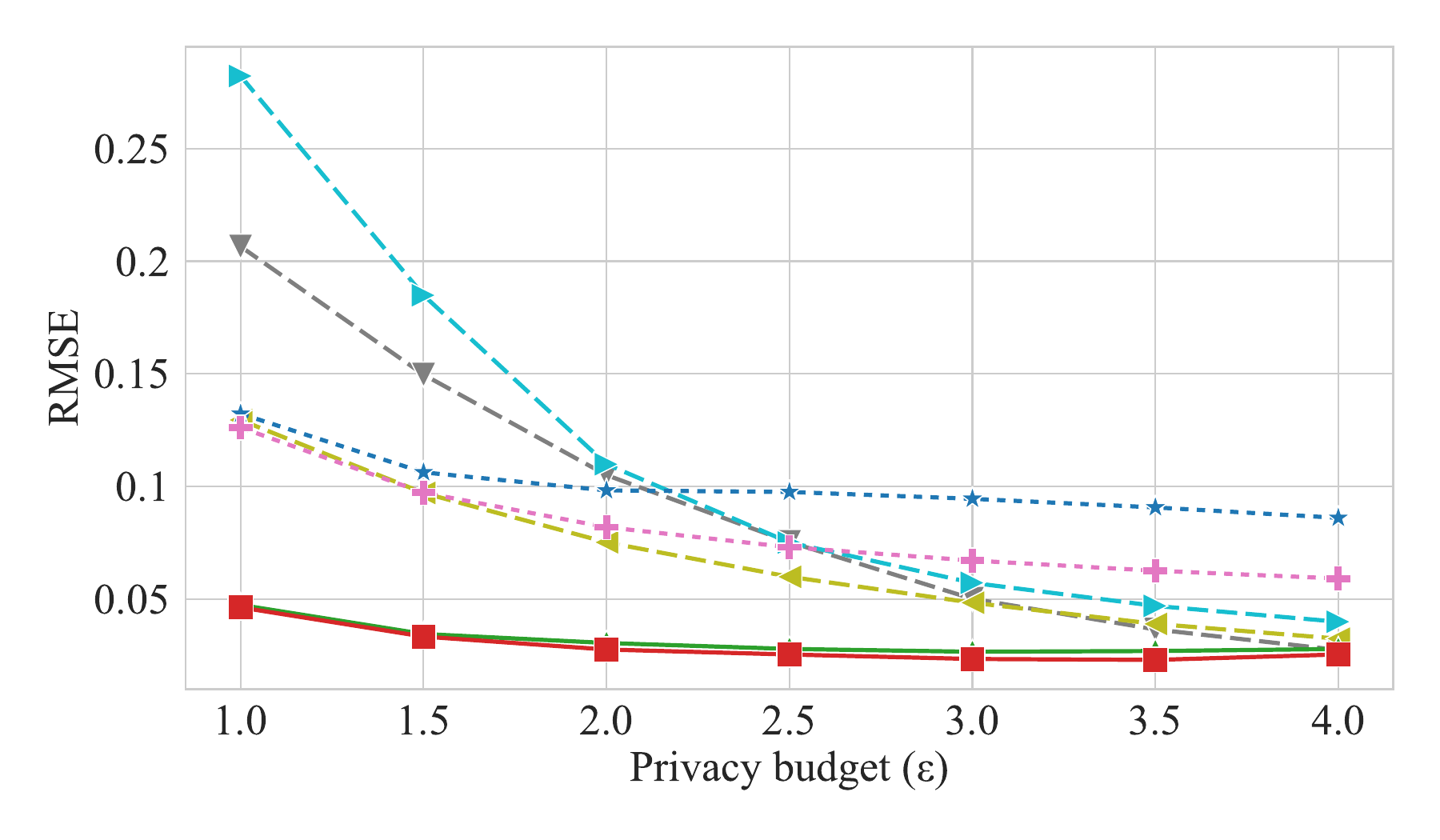} 
        \caption{Adult}
        \label{fig:adult_var}
    \end{subfigure}
    \begin{subfigure}[b]{0.497\textwidth}
        \centering
        \includegraphics[width=\textwidth, trim={1.0cm 1.0cm 1.0cm 1.0cm}, clip]{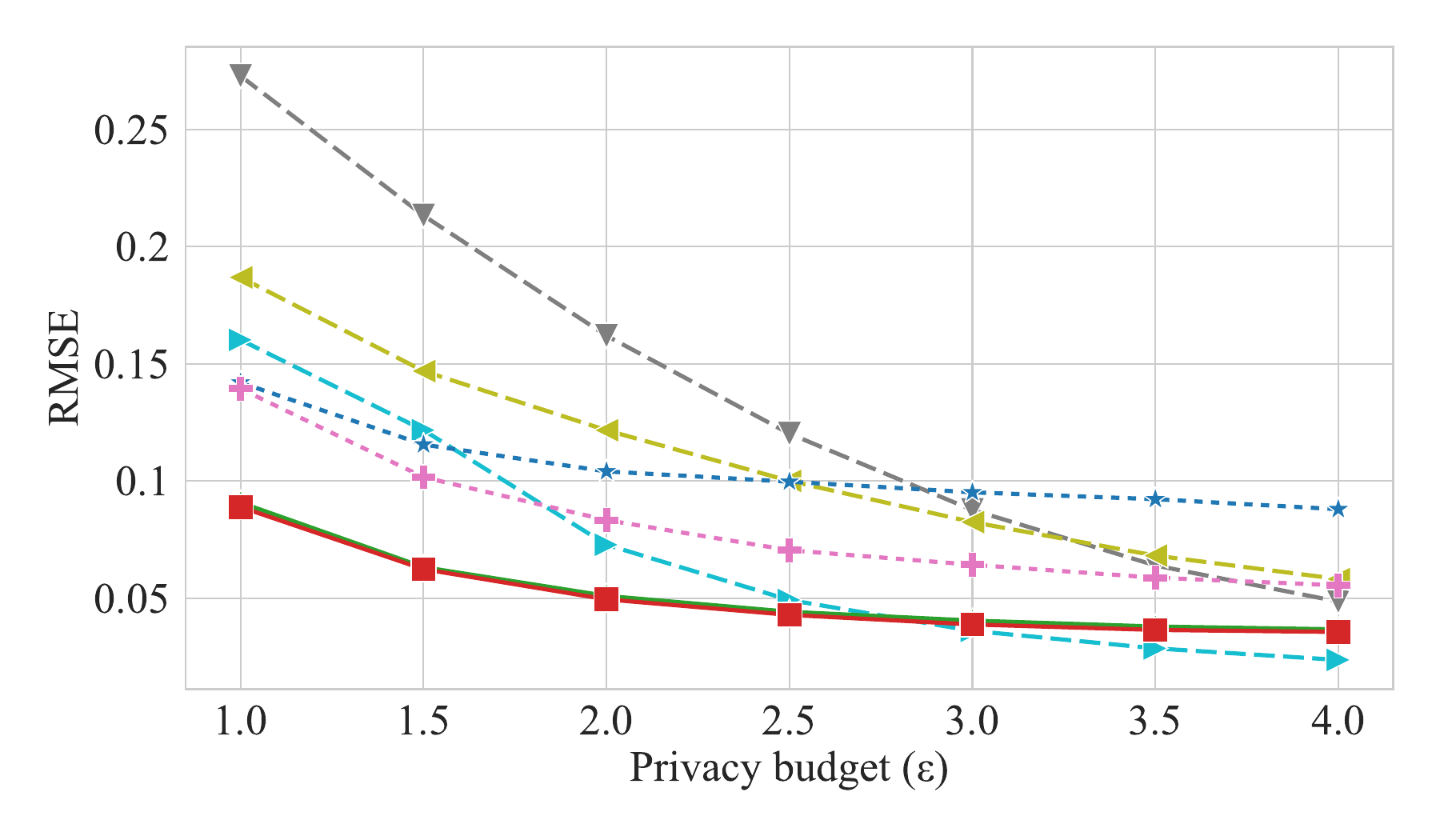} 
        \caption{Credit}
        \label{fig:credit_var}
    \end{subfigure}
    \begin{subfigure}[b]{0.497\textwidth}
        \centering
        \includegraphics[width=\textwidth, trim={1.0cm 1.0cm 1.0cm 1.0cm}, clip]{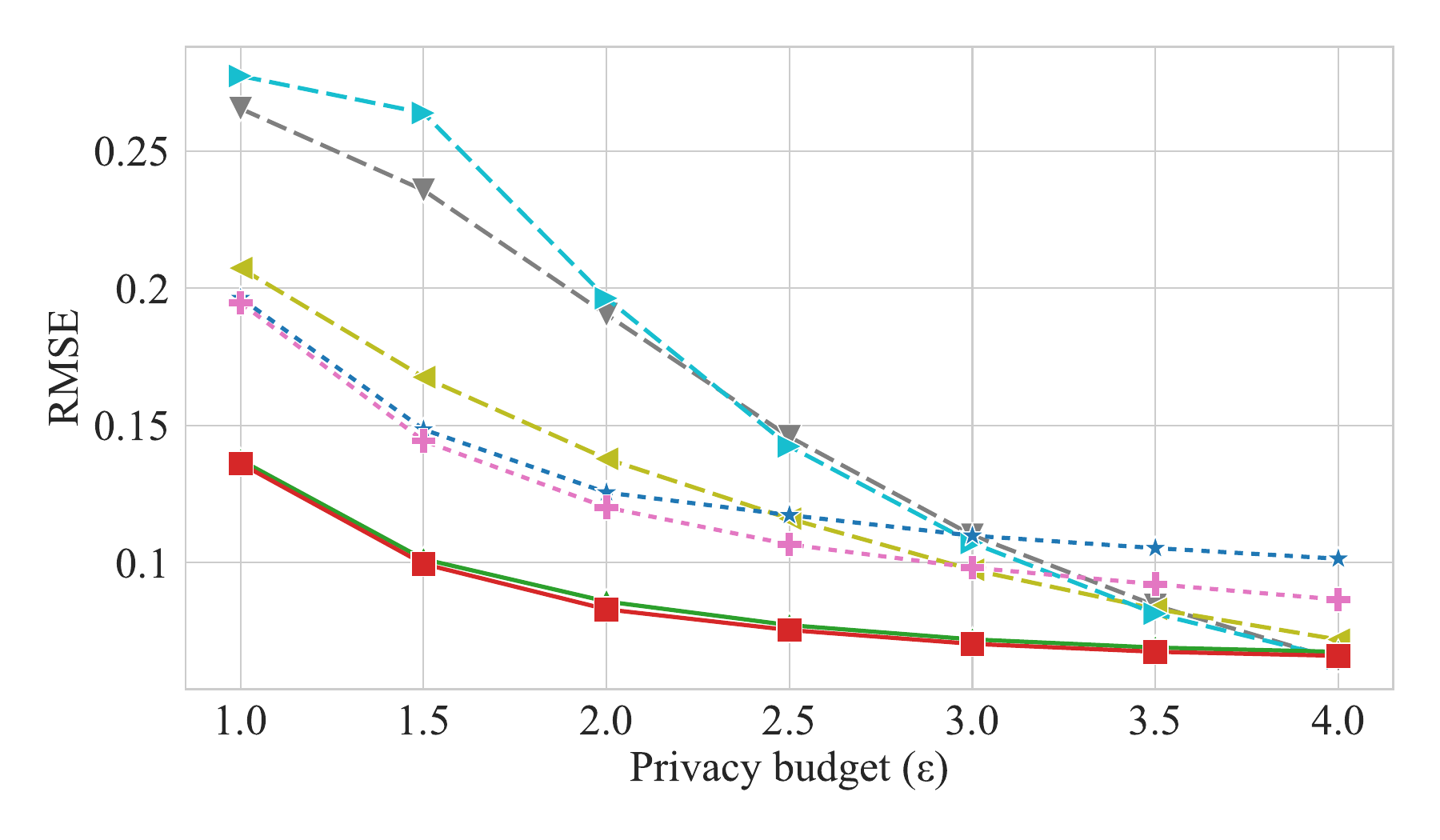} 
        \caption{Online News Popularity}
        \label{fig:ONP_var}
    \end{subfigure}
    \begin{subfigure}[b]{0.497\textwidth}
        \centering
        \includegraphics[width=\textwidth, trim={1.0cm 1.0cm 1.0cm 1.0cm}, clip]{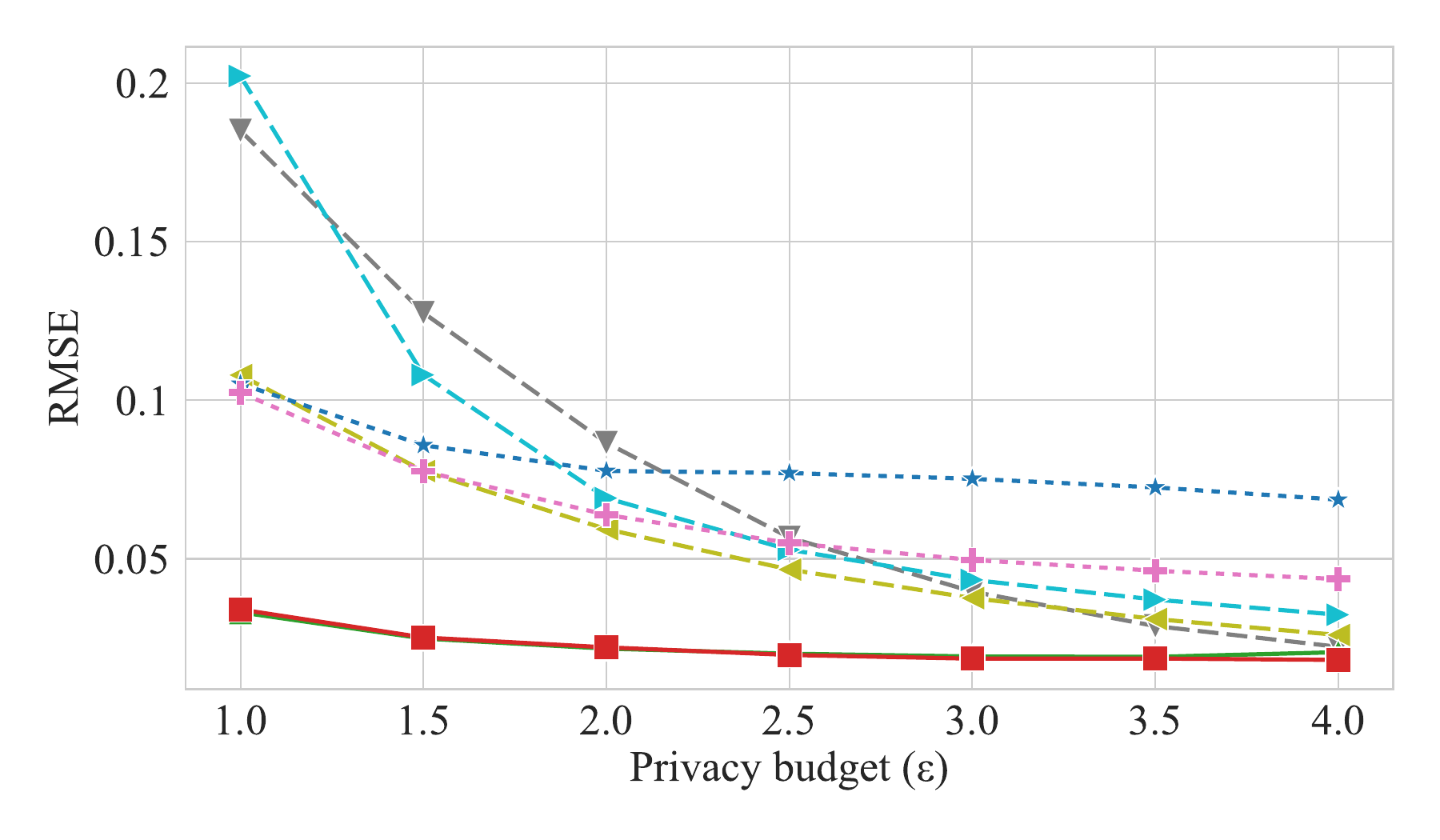} 
        \caption{Beijing Housing}
        \label{fig:BHP_var}
    \end{subfigure}
    \caption{Variance estimation (RMSE) on real-world datasets. Lower is better.}
    \label{fig:exp_var}
\end{figure*}

\subsection{Distribution estimation}~\label{sec:dist_est}
We now evaluate the accuracy of distribution estimation using the Wasserstein distance. The Wasserstein distance, also known as the Earth Mover's Distance, measures the minimum cost of transforming one distribution into another. It is formally defined as:
\begin{equation}
    W(\hat{\pi}, \pi) = \int_{-\infty}^{\infty} |F_{\hat{\pi}}(x) - F_{\pi}(x)| dx
\end{equation}
where $F_{\hat{\pi}}$ and $F_{\pi}$ are the cumulative distribution functions of the estimated and true distributions, respectively. To estimate the distribution, we partitioned the normalized domain $[-1,1]$ into $d=64$ uniform bins and applied the EM framework detailed in Section~\ref{sec:em}, using a convergence threshold of $\tau=1 \times 10^{-5}$. For this task, we include the SW~\cite{li2020estimating} mechanism, which is designed for distribution estimation, and two strong categorical baselines, DE and OUE~\cite{wang2017locally}. For graph clarity, we omit Duchi's, Three-output, and our N-output (Analytical) mechanism, as their small $N$ makes them unsuitable for this task.

The results, shown in Figure~\ref{fig:exp_wass}, demonstrate that both the N-output (Worst) and (Avg) mechanisms outperform the baseline methods in most cases. The performance gap can be substantial, with our mechanisms being up to five times more accurate in certain settings. Similar to the mean estimation task, our mechanisms hold a significant advantage at lower $\epsilon$ values. In the large $\epsilon$, the performance of PMs, SW, DE, and OUE improves, gradually closing the gap, but the N-output mechanism maintains its overall superiority. An interesting finding is that the N-output (Avg) mechanism generally yields slightly better accuracy than N-output (Worst) while N-output (Worst) is slightly better than N-output (Avg).

\subsection{Variance estimation}
Finally, we assess the performance of variance estimation. Once the underlying data distribution is reconstructed, as evaluated in Section~\ref{sec:dist_est}, it becomes possible to derive various other statistical moments. The variance is calculated as $\mathbb{E}[X^2] - \mathbb{E}[X]^2$. The term $\mathbb{E}[X^2]$ is computed as the weighted average of the squared value of each bin's center, where the weights are the estimated frequencies for each bin. For PMs and N-output mechanisms, we use the EM-based framework developed in Section \ref{sec:em}.

Figure~\ref{fig:exp_var} shows that both N-output (Worst) and (Avg) consistently outperform other mechanisms across nearly all settings. The performance gap is particularly pronounced in the high privacy regime, where our mechanisms are at least 1.5 times, and up to 4 times, more accurate than the baselines. As $\epsilon$ increases, the accuracy of continuous-output (PMs, SW) and categorical (DE, OUE) mechanisms gradually approaches that of the N-output mechanism. However, our proposed mechanisms maintain their overall superiority, with the minor exception of SW on the Credit dataset for $\epsilon \geq 3.0$.

\section{Conclusion}
In this work, we addressed a significant gap in the literature for numerical data collection under LDP: the absence of a generalized method for constructing an optimal mechanism for an arbitrary, finite output size $N$. We proposed the \textbf{N-output mechanism}, a novel and adaptive framework that fills this gap by formulating the mechanism's design as a solvable optimization problem. By developing both highly accurate numerical and theoretically grounded analytical solutions, our framework generalizes prior work into a single, flexible, and powerful mechanism. Furthermore, we advanced the field by developing a unified EM-based framework that enables distribution estimation for existing numerical LDP mechanisms, a previously unaddressed capability.

Our empirical evaluation demonstrates that the N-output mechanism achieves SOTA accuracy for mean, variance, and distribution estimation, outperforming existing mechanisms across a wide range of privacy settings. This superior performance is achieved with minimal communication overhead, making it up to 16 times more bit-efficient than continuous-output baselines. The N-output mechanism stands as a robust and practical solution, offering high utility and communication efficiency for privacy-preserving numerical data collection.
  

\cleardoublepage
\raggedbottom
\appendix
\bibliographystyle{plain}
\bibliography{Noutput}

\begin{thebibliography}{10}

\bibitem{bassily2015local}
Raef Bassily and Adam Smith.
\newblock Local, private, efficient protocols for succinct histograms.
\newblock In {\em Proceedings of the forty-seventh annual ACM symposium on Theory of computing}, pages 127--135, 2015.

\bibitem{adult_2}
Barry Becker and Ronny Kohavi.
\newblock {Adult}.
\newblock UCI Machine Learning Repository, 1996.
\newblock {DOI}: https://doi.org/10.24432/C5XW20.

\bibitem{dempster1977maximum}
Arthur~P Dempster, Nan~M Laird, and Donald~B Rubin.
\newblock Maximum likelihood from incomplete data via the em algorithm.
\newblock {\em Journal of the royal statistical society: series B (methodological)}, 39(1):1--22, 1977.

\bibitem{ding2017collecting}
Bolin Ding, Janardhan Kulkarni, and Sergey Yekhanin.
\newblock Collecting telemetry data privately.
\newblock {\em Advances in Neural Information Processing Systems}, 30, 2017.

\bibitem{duchi2018minimax}
John~C Duchi, Michael~I Jordan, and Martin~J Wainwright.
\newblock Minimax optimal procedures for locally private estimation.
\newblock {\em Journal of the American Statistical Association}, 113(521):182--201, 2018.

\bibitem{dwork2006differential}
Cynthia Dwork.
\newblock Differential privacy.
\newblock In {\em International colloquium on automata, languages, and programming}, pages 1--12. Springer, 2006.

\bibitem{erlingsson2014rappor}
{\'U}lfar Erlingsson, Vasyl Pihur, and Aleksandra Korolova.
\newblock Rappor: Randomized aggregatable privacy-preserving ordinal response.
\newblock In {\em Proceedings of the 2014 ACM SIGSAC conference on computer and communications security}, pages 1054--1067, 2014.

\bibitem{fanti2016building}
Giulia Fanti, Vasyl Pihur, and {\'U}lfar Erlingsson.
\newblock Building a rappor with the unknown: Privacy-preserving learning of associations and data dictionaries.
\newblock {\em Proceedings on Privacy Enhancing Technologies}, 2016.

\bibitem{online_news_popularity_332}
Kelwin Fernandes, Pedro Vinagre, Paulo Cortez, and Pedro Sernadela.
\newblock {Online News Popularity}.
\newblock UCI Machine Learning Repository, 2015.
\newblock {DOI}: https://doi.org/10.24432/C5NS3V.

\bibitem{kasiviswanathan2011can}
Shiva~Prasad Kasiviswanathan, Homin~K Lee, Kobbi Nissim, Sofya Raskhodnikova, and Adam Smith.
\newblock What can we learn privately?
\newblock {\em SIAM Journal on Computing}, 40(3):793--826, 2011.

\bibitem{li2020estimating}
Zitao Li, Tianhao Wang, Milan Lopuha{\"a}-Zwakenberg, Ninghui Li, and Boris {\v{S}}koric.
\newblock Estimating numerical distributions under local differential privacy.
\newblock In {\em Proceedings of the 2020 ACM SIGMOD International Conference on Management of Data}, pages 621--635, 2020.

\bibitem{ruiqurm_2018}
Ruiqurm.
\newblock Beijing lianjia housing prices, 2018.

\bibitem{Apple}
Apple Differential~Privacy Team.
\newblock Learning with privacy at scale, 2017.

\bibitem{wang2019collecting}
Ning Wang, Xiaokui Xiao, Yin Yang, Jun Zhao, Siu~Cheung Hui, Hyejin Shin, Junbum Shin, and Ge~Yu.
\newblock Collecting and analyzing multidimensional data with local differential privacy.
\newblock In {\em 2019 IEEE 35th International Conference on Data Engineering (ICDE)}, pages 638--649. IEEE, 2019.

\bibitem{wang2017locally}
Tianhao Wang, Jeremiah Blocki, Ninghui Li, and Somesh Jha.
\newblock Locally differentially private protocols for frequency estimation.
\newblock In {\em 26th USENIX Security Symposium (USENIX Security 17)}, pages 729--745, 2017.

\bibitem{wang2019consistent}
Tianhao Wang, Milan Lopuha{\"a}-Zwakenberg, Zitao Li, Boris Skoric, and Ninghui Li.
\newblock Consistent and accurate frequency oracles under local differential privacy.
\newblock {\em arXiv}, 2019:1905--08320, 2019.

\bibitem{warner1965randomized}
Stanley~L Warner.
\newblock Randomized response: A survey technique for eliminating evasive answer bias.
\newblock {\em Journal of the American statistical association}, 60(309):63--69, 1965.

\bibitem{default_of_credit_card_clients_350}
I-Cheng Yeh.
\newblock {Default of Credit Card Clients}.
\newblock UCI Machine Learning Repository, 2009.
\newblock {DOI}: https://doi.org/10.24432/C55S3H.

\bibitem{zhao2020local}
Yang Zhao, Jun Zhao, Mengmeng Yang, Teng Wang, Ning Wang, Lingjuan Lyu, Dusit Niyato, and Kwok-Yan Lam.
\newblock Local differential privacy-based federated learning for internet of things.
\newblock {\em IEEE Internet of Things Journal}, 8(11):8836--8853, 2020.

\end{thebibliography}

\cleardoublepage
\section{Additional Preliminaries} \label{sec:add_prelim}
The Duchi's mechanism randomly produces one of the two possible outputs ($|\Omega|=2$). The output values are constants, which are determined by the privacy budget $\epsilon$. The probabilities of reporting these output values depend on the true value $x$ and are defined as follows:
\begin{equation*}
    Pr[y|x]=\begin{cases} 
      \frac{e^\epsilon-1}{2e^\epsilon+2}x + \frac{1}{2}, & \text{if } y=\frac{e^\epsilon +1}{e^\epsilon-1}, \\[0.5em]
      -\frac{e^\epsilon-1}{2e^\epsilon+2}x + \frac{1}{2}, & \text{if } y=-\frac{e^\epsilon +1}{e^\epsilon-1}.
   \end{cases}
\end{equation*}
Similarly, the three-output mechanism \cite{zhao2020local} randomly produces one of three possible outputs ($|\Omega|=3$).

On the other hand, PMs randomly produce a real number with $\Omega = [L(-1), R(1)]$ which is defined as \eqref{PM_def}. The initial version of PM is proposed by~\cite{wang2019collecting}, and later, Zhao et al.~\cite{zhao2020local} proposed both optimal and sub-optimal variants. The probability density function is given by:
\begin{equation} \label{PM_def}
    f[y|x]=\begin{cases} 
      e^\epsilon q & \text{if } y \in [L(x), R(x)],\ \\[0.5em]
      q & \text{if } y \in [L(-1), L(x)) \cup (R(x), R(1)]
   \end{cases}
\end{equation}
where $q, L(x)$, and $R(x)$ are defined as follows:
\begin{equation}
\begin{aligned}
    q &= \frac{\zeta(e^\epsilon-1)}{2(e^\epsilon+\zeta)^2}\\
    L(x) &= \frac{(e^\epsilon + \zeta)(x\zeta-1)}{\zeta(e^\epsilon-1)}\\
   R(x) &= \frac{(e^\epsilon + \zeta)(x\zeta+1)}{\zeta(e^\epsilon-1)}
\end{aligned}
\end{equation}
The parameter $\zeta$ shapes the characteristics of PMs. Specifically, when $\zeta=e^{\epsilon/2}$, it becomes the original PM. For $\zeta=e^{\epsilon/3}$, it becomes the sub-optimal variant, PM-sub. For PM-opt, since it involves complex formulation for $\zeta$, refer to~\cite{zhao2020local}.

\section{Deferred Proofs}
\subsection{Proof for Proposition \ref{prop:1}} \label{pf:prop:1}
\begin{proof}
For an input $x$ within an interval $[x_{j-1}, x_j]$, $\Var[Y|x]$ is downward-opening parabola. The maximum value of such a parabola on a closed interval $[x_{j-1}, x_j]$ must occur at one of the endpoints ($x_{j-1}$ or $x_j$) or at its vertex, if the vertex lies within the interval.

Within the interval $[x_{j-1}, x_j]$, $\Var[Y|x]$ can be written as:
\begin{equation}
    \Var[Y|x] = -x^2 + \frac{\sum_{i=-n}^n(P_{i, j} - P_{i,j})a_i^2 }{x_j - x_{j-1}} x + C
\end{equation}
where $C$ is constant. The vertex of this parabola is:
\begin{equation}
    \bar{x}_j = \frac{\sum_{i=-n}^n (P_{i,j} - P_{i,j-1}) a_i^2}{2(x_j - x_{j-1})}
\end{equation}
where $x_j = \sum_{i=-n}^n a_iP_{i,j}$. 
\end{proof}

\subsection{Proof for Proposition \ref{prop:2}} \label{pf:prop:2}
\begin{proof}
For any input $x$ within an interval $[x_{j-1}, x_j]$, the expected output values is $\mathbb{E}[Y|x] = \sum_{i=-n}^n a_ip_{i,j}(x)$. The equation can be expressed as:
\begin{equation} \label{E}
\mathbb{E}[Y|x] = \frac{\sum_{i=-n}^{n}(P_{i, j} - P_{i, j-1})a_i }{x_j - x_{j-1}} (x - x_{j}) + \sum_{i=-n}^{n}a_iP_{i,j}
\end{equation}

\noindent By the definition $x_j = \sum_{i=-n}^n a_i P_{i,j}$, Eq.~\eqref{E} becomes:
\begin{equation}
\begin{split}
    \mathbb{E}[Y|x] &= \frac{ x_j - x_{j-1} }{x_j - x_{j-1}} (x - x_{j}) + x_j \\
    &=x
\end{split}
\end{equation}
\end{proof}

\subsection{Proof for Proposition \ref{prop:3}} \label{pf:prop:3}
\begin{proof}
The noise variance is the expected squared error: $\Var[Y|x] = \mathbb{E}[(Y-x)^2|x]$. To minimize the variance, the mechanism must assign probability mass efficiently. It should assign higher probability $p_i(x)$ where the squared error $(a_i - x)^2$ is small, and lower probability where the error is large.

Consider a fixed output $a_i >0$. The squared error is largest when the input $x$ is farthest away, $x=-1$. To construct an optimal mechanism that keeps variance low across the domain, it must assign the lowest possible probability to $p_i(x)$ at this point. Due to the symmetric $\Omega$, the error is largest at $x=1$ for $a_i < 0$.
\end{proof}

\subsection{Proof for Proposition \ref{prop:4}} \label{pf:prop:4}
\begin{proof}
    This proof has two parts: non-negativity and sum-to-one. For any $x \in [x_{j-1}, x_j]$, $p_{i,j}(x)$ can be written as:
    \begin{equation}
    \begin{split}
        &p_{i,j}(x) = \frac{P_{i, j} - P_{i, j-1} }{x_j - x_{j-1}} (x - x_{j}) + P_{i,j}\\
         &= \left( \frac{x-x_j}{x_j - x_{j-1}} \right)P_{i,j} - \left( \frac{x-x_j}{x_j - x_{j-1}} \right)P_{i,j-1} + P_{i,j}\\
         & = \left( \frac{x_j - x}{x_j - x_{j-1}} \right)P_{i,j-1} + \left( \frac{x-x_{j-1}}{x_j - x_{j-1}} \right)P_{i,j}
    \end{split}
    \end{equation}
    Then, $p_{i,j}(x)$ is a linear interpolation of the endpoint probabilities:
    \begin{equation}
        p_{i,j}(x) = \alpha P_{i,j-1} + (1-\alpha)P_{i,j}
    \end{equation}
    where $\alpha = \frac{x_j - x}{x_j - x_{j-1}}$.
    
    \textbf{Non-negativity.} By premise, the endpoint probabilities are non-negative, $P_{i,j-1} \geq 0$ and $P_{i,j} \geq 0$. As $\alpha$ and $(1-\alpha)$ are also non-negative, $p_i(x)$ is a non-negative weighted sum of non-negative values. Thus, $p_i(x) \geq 0$ for all $x$.

    \textbf{Sum-to-one.} The sum of probabilities is:
    \begin{equation}
    \begin{split}
        \sum_{i=-n}^n p_i(x) &= \sum_{i=-n}^n (\alpha P_{i, j-1} + (1-\alpha)P_{i,j}) \\
        & = \alpha \left( \sum_{i=-n}^n P_{i,j-1} \right) + (1-\alpha)\left( \sum_{i=-n}^n P_{i,j} \right)
    \end{split}
    \end{equation}
    By premise, the endpoint probabilities sum to one: $\sum_i P_{i, j-1} = 1$ and $\sum_i P_{i, j} = 1$. Substituting these values gives
    \begin{equation}
        \sum_{i=-n}^n p_i(x) = (1-\alpha) \cdot 1 + \alpha \cdot 1 = 1   
    \end{equation}
\end{proof}

\subsection{Proof for Lemma \ref{lem:max_point}} \label{pf:lem:max_point}
    \begin{proof}
     For $j \geq 2$, we consider the following three cases regarding the maximum values:
        \begin{enumerate}
            \item If $a_{j-1} < 2x_{j-1} - a_j$, $\text{Var}_j[Y|x]$ attains its maximum value at $x=x_{j-1}$ since $x_{j-1} > x^*_j$.
                    
            \item If $2x_{j-1} - a_j \leq a_{j-1} \leq 2x_j - a_j$, $\text{Var}_j[Y|x]$ attains its maximum value at $x=x^*_j$ since $x_{j-1} \leq x^*_j \leq x_j$.
                
            \item If $2x_j - a_j < a_{j-1}$, $\text{Var}_j[Y|x]$ attains its maximum value at $x=x_i$ since $x_j < x^*_j$.
        \end{enumerate}
        To prove this lemma, we show that case (1) is not feasible and case (3) does not minimize the worst-case noise variance.
        
        Case (1): Rewriting the condition $a_{j-1} < 2x_{j-1} - a_j$, given that $x_{j-1} = \frac{a_{j-1}}{a_n}$, we obtain:
        \begin{equation*}
            a_na_{j-1} + a_na_j \leq 2a_{j-1}
        \end{equation*}
        This inequality does not hold since $a_na_{j-1} > a_{j-1}$ and $a_n a_j > a_{j-1}$. Hence, case (1) is not feasible.
        
        Case (3): Let $k \in \{2,\ldots,n\}$ and $l \in \{1, \ldots, n \}$ where $l \neq k$. We have $\frac{\partial \text{Var}_k[Y|x_k]}{\partial a_{k-1}} > 0$ and $\frac{\partial \text{Var}_l[Y|x]}{\partial a_{k-1}} > 0$. These indicate that decreasing the value of $a_{k-1}$ decreases every $\max_x \text{Var}_j[Y|x]$ for any $j \geq 2$. To ensure that the worst-case noise variance is minimized, $a_{j-1} \leq (2t-1)a_j$ must hold. Thus, $x^*_j$ is the point that maximizes $Var_j[Y|x]$ to minimize the worst-case noise variance.
    \end{proof}

\subsection{Proof for Lemma \ref{lem:worst-case}} \label{pf:lem:worst-case}
\begin{proof}
        To prove this, we assume for contradiction that $\text{Var}_n[Y|x^*_n]$ is not the worst-case noise variance and show that, under this assumption, the worst-case noise variance cannot be minimized. Suppose there exists some $j < n$ such that $\text{Var}_j[Y|x^*_j]$ is the worst-case noise variance, while $\text{Var}_n[Y|x^*_n]$ is not. This implies $\text{Var}_j[Y|x^*_j]$ > $\text{Var}_n[Y|x^*_n]$. 

        First, consider the derivative of $\text{Var}_j[Y|x^*_j]$ with respect to $a_{n-1}$ for $j \in \{2,\ldots,n-2\}$:
        \begin{equation*}
            \frac{\partial \text{Var}_i[Y|x^*_i]}{\partial a_{n-1}} = 4pa_{n-1} > 0    
        \end{equation*}
        This positive derivative indicates that decreasing $a_{n-1}$ decreases the worst-case noise variance.
        
        Next, for $j=n-1$, we investigate $\frac{\partial \text{Var}_{n-1}[Y|x^*_{n-1}]}{\partial a_{n-1}} > 0$ which is equivalent to:
        \begin{equation} \label{app:pf:eq:wc1}
            (8p + 1)a_{n-1} > (2(e^\epsilon-1)p - 1) a_{n-2}
        \end{equation}
        We consider two cases:
        \begin{enumerate}
            \item For $(2(e^\epsilon-1)p - 1) \leq 0$, inequality \eqref{app:pf:eq:wc1} is trivially satisfied since the right-hand side is non-positive.
            \item For $(2(e^\epsilon-1)p - 1) > 0$, rearranging inequality \eqref{app:pf:eq:wc1}, we obtain:
            \begin{equation*}
                \frac{8p + 1}{2(e^\epsilon-1)p - 1}a_{n-1} > a_{n-2}
            \end{equation*}
            As the coefficient of $a_{n-1}$ decrease with $p$ increases, the inequality holds if it is satisfied for $p = \frac{1}{e^\epsilon + 2n -1}$. Substituting this value of $p$ gives:
            \begin{equation*}
                \frac{e^\epsilon + 2n +7}{e^\epsilon - 2n - 1}a_{n-1} > a_{n-2} 
            \end{equation*}
            which confirms that the inequality is satisfied.
        \end{enumerate}
        These results imply that decreasing $a_{n-1}$ reduces the worst-case noise variance $\text{Var}_{n-1}[Y|x^*_{n-1}]$ for any $j \geq 2$. 
                
        Next, we examine the derivative of $\text{Var}_{n}[Y|x^*_{n}]$ with respect to $a_{n-1}$:
        \begin{align*}
            \frac{\partial \text{Var}_{n}[Y \mid x^*_{n}]}{\partial a_{n-1}} = (4p + \frac{1}{2})a_{n-1} - (t - \frac{1}{2})a_{n}
        \end{align*}
        leading to two scenarios:
        \begin{equation*}
            \begin{cases}
                \frac{\partial \text{Var}_{n}[Y \mid x^*_{n}]}{\partial a_{n-1}} > 0 & \text{if } a_{n-1} > \frac{e^\epsilon -N -1}{e^\epsilon + N + 7}a_n \\[0.5em]
                \frac{\partial \text{Var}_{n}[Y \mid x^*_{n}]}{\partial a_{n-1}} < 0 & \text{if } a_{n-1} < \frac{e^\epsilon -N -1}{e^\epsilon + N + 7}a_n
            \end{cases}
        \end{equation*}
        Let the point $z = \frac{e^\epsilon - 2n -1}{e^\epsilon + 2n + 7}a_n$.
        For $a_{n-1} > z$, $a_{n-1}$ must be reduced at least to $z$ to minimize the worst-case noise variance.
        For $a_{n-1} \leq z$, $a_{n-1}$ can be reduced until $\text{Var}_n[Y|x^*_n]$ becomes the worst-case noise variance, thereby minimizing the worst-case noise variance. Therefore, $\text{Var}_n[Y|x^*_n]$ must be the worst-case noise variance and should be minimized.
    \end{proof}

\subsection{Proof for Theorem \ref{thm:type0}} \label{pf:thm:type0}
\begin{proof}
        Suppose $a_{n-1}$ is given. Reducing $a_{n-2}$ decreases $\text{Var}_n[Y|x^*_n]$ and generally increases $\text{Var}_{n-1}[Y|x^*_{n-1}]$. Hence, the value of $a_{n-2}$ such that $\text{Var}_{n-1}[Y|x^*_{n-1}] = \text{Var}_n[Y|x^*_n]$ maintains $\text{Var}_n[Y|x^*_n]$ as the worst-case noise variance and can minimize it.
        Applying the same logic to $a_{n-3}, \ldots, a_1$, the values of $a_j$ such that $\text{Var}_n[Y|x^*_n] = \text{Var}_j[Y|x^*_j]$ for $j \in \{2, 3, \ldots, n-1\}$ can minimize the worst-case noise variance.
        
        Given $a_{j+1}, a_{j+2}$ where $j=n-2,\ldots,1$, the equation that we solve to get $a_j$ is:
        \begin{equation*}
            \text{Var}_{j+1}[Y \mid x^*_{j+1}] = \text{Var}_{j+2}[Y \mid x^*_{j+2}]
        \end{equation*}
        which is equivalent to
        \begin{equation*}
        \begin{split}
            \frac{(a_j + a_{j+1})^2}{4} - ta_ja_{j+1}= \frac{(a_{j+1} + a_{j+2})^2}{4} - ta_{j+1}a_{j+2}
        \end{split}
        \end{equation*}
        Rewriting for $a_j$, we have:
        \begin{equation*}
        a_j^2 + 2a_{j+1}(1 - 2t)a_j + a^2_{j+1} -(a_{j+1} + a_{j+2})^2 + 4ta_{j+1}a_{j+2} = 0
        \end{equation*}
        Solving this with the quadratic formula, we have:
        \begin{equation} \label{eq:ak}
            a_j = a_{j+2} \quad \text{or} \quad a_j = (4t - 2)a_{j+1} - a_{j+2}
        \end{equation}
        Obviously, $a_j = (4t - 2)a_{j+1} - a_{j+2}$. By this recurrence relation, $a_j$ can be represented by $a_{n-1}$ and $a_n$. We define:
        \begin{equation} \label{eq:ak2}
            a_j = T_j a_{n-1} + Q_j a_n
        \end{equation}
        where $T_j$ and $Q_j$ are coefficient of $a_{n-1}$ and $a_n$, respectively. When $j=n$, $T_n = 0$, $Q_n =1$. When $j=n-1$, $T_{n-1}=1$, $Q_{n-1} = 0$. Plugging \eqref{eq:ak2} to \eqref{eq:ak}, the sequence $T_j$ and $Q_j$ are defined as:
        \begin{gather*}
            T_j = (4t - 2)T_{j+1} - T_{j+2} \quad \text{for } j \in \{1, 2, \ldots, n-2\}\\[0.5em]
            Q_j = (4t - 2)Q_{j+1} - Q_{j+2} \quad \text{for } j \in \{1, 2, \ldots, n-2\}
        \end{gather*}
        Then, $\text{Var}_n[Y|x^*_n]$ can be represented as:
        \begin{equation*}
            \text{Var}_n[Y \mid x^*_n] = \frac{(a_{n-1} + a_n)^2}{4} - a_{n-1} + 2p\sum_{i=1}^n(T_i a_{n-1} + Q_i a_n)^2
        \end{equation*}
        since it is convex quadratic $a_{n-1}$,
        \begin{equation*}
            a_{n-1} = \frac{(2t - 1) - 8p \sum_{i=1}^n T_i Q_i}{1 + 8p \sum_{i=1}^n T_i^2} a_n
        \end{equation*}
        minimizes $\text{Var}_n[Y|x^*_n]$ maintaining it as the worst-case noise variance only if
        \begin{gather*}
        a_j < a_{j+1} \quad \text{for } j \in \{0, 1, \ldots, n-1\}\\
        \text{Var}_n[Y \mid x^*_n; p_0] \geq \text{Var}_1[Y \mid x^*_1; p_0]
    \end{gather*}
    \end{proof}

\subsection{Proof for Theorem \ref{thm:type1}} \label{pf:thm:type1}
\begin{proof}
        Extending the proof of Theorem \ref{thm:type0}, we maintain that $\text{Var}_j[Y|x^*_j] = \text{Var}_n[Y|x^*_n]$ for $j \in \{1, 2, \ldots, n-1\}$. This ensures that $\text{Var}_n[Y|x^*_n]$ remains the worst-case noise variance and minimizes it. 
    Therefore, we need to solve the following equation:
    \begin{equation*}
        \text{Var}_j[Y|x^*_j] = \text{Var}_{j+1}[Y|x^*_{j+1}] \quad \text{for } j \in \{1, 2, \ldots, n-1\}
    \end{equation*}
    Solving $\text{Var}_1[Y|x^*_1] = \text{Var}_2[Y|x^*_2]$:
    \begin{equation*}
    a_1 = 
        \begin{cases}
            \frac{a_2}{4t-1} \quad \text{for } N=2n\\
            \frac{a_2}{4t-2} \quad \text{for } N=2n+1
        \end{cases}
    \end{equation*}
    Let $a_j = C_ja_{j+1}$, then we can define $C_1$ as:
    \begin{equation} \label{eq:c1}
        C_1 = \begin{cases}
             \frac{1}{4t-1} & \text{if } N=2n \\
             \frac{1}{4t-2} & \text{if } N=2n+1 \\
        \end{cases}
    \end{equation}
    Next, we solve $\text{Var}_j[Y|x^*_j] = \text{Var}_{j+1}[Y|x^*_{j+1}]$:
    \begin{equation*}
        (a_{j-1} + a_j)^2 -4ta_{j-1}a_j = (a_j + a_{j+1})^2 - 4ta_ja_{j+1}
    \end{equation*}
    Rewriting this using $a_{j-1} = C_{j-1}a_j$, we get:
    \begin{align*}
        (C^2_{j-1} + 2C_{j-1} -4tC_{j-1})a^2_j + (4ta_{j+1} - 2a_{j+1})a_j - a^2_{j+1} = 0 
    \end{align*}
    Solving this for $a_j$, we get:
    \begin{equation*}
        a_j = \frac{1-2t \pm \sqrt{C^2_{j-1} + 2C_{j-1} -4tC_{j-1} + (2t-1)^2}}{C^2_{j-1} + 2C_{j-1} -4tC_{j-1}}a_{j+1}
    \end{equation*}
    Since we require $0 < C_j < 1$, we have:
    \begin{equation*}
        C_{j} = \frac{1-2t + \sqrt{C^2_{j-1} + 2C_{j-1} -4tC_{j-1} + (2t-1)^2}}{C^2_{j-1} + 2C_{j-1} -4tC_{j-1}}
    \end{equation*}
    \end{proof}

\subsection{Proof for Lemma \ref{lem:var_infty}} \label{pf:lem:var_infty}
\begin{proof}
    Given $N$, the worst-case noise variance of the N-output mechanisms decreases when $\epsilon$ increases. 
    Thus, we need to show that:
    \begin{equation*}
        \lim_{\epsilon\to\infty} 2\sum_{i=1}^{n}a_i^2p + \frac{(a_{n-1} + a_n)^2}{4} - a_{n-1} = \frac{1}{(N-1)^2}
    \end{equation*}
    First, we examine the limit of the term $\lim_{\epsilon\to\infty} 2\sum_{i=1}^{n}a_i^2p$.
    \begin{equation*}
        \lim_{\epsilon\to\infty} 2\sum_{i=1}^{n}a_n^2p = \frac{2n(e^\epsilon+N-1)}{(e^\epsilon -1)^2} = 0
    \end{equation*}
    Since $ \sum_{i=1}^{n}a_n^2p > \sum_{i=1}^{n}a_i^2p > 0$, by the squeeze theorem, we have:
    \begin{equation*}
        \lim_{\epsilon\to\infty} 2\sum_{i=1}^{n}a_i^2p = 0
    \end{equation*}
    Next, we need to find $\lim_{\epsilon \to \infty} a_{n-1}$.
    Since $a_{n-1} = C_{n-1}a_n$, we need to determine $\lim_{\epsilon \to \infty} C_{n-1}$.
    \begin{equation} \label{infty_Ck}
        \begin{split}
        \quad  & \lim_{\epsilon\to\infty} C_{i+1}\\
        = &\lim_{\epsilon\to\infty} \frac{-1 + \sqrt{(C_i - 1)^2}}{C^2_i - 2C_i} \quad (\because \lim_{\epsilon\to\infty} t)\\
        = &\lim_{\epsilon\to\infty} \frac{1}{2 - C_i} \quad (\because C_i < 1)
        \end{split}
    \end{equation}
    By \eqref{eq:c1} and \eqref{infty_Ck}, we have:
    \begin{equation*}
        \begin{cases}
            \lim_{\epsilon\to\infty} C_i = \frac{2i - 1}{2i + 1} & \text{ if } N = 2n \\
            \lim_{\epsilon\to\infty} C_i = \frac{i}{i + 1} & \text{ if } N = 2n + 1
        \end{cases}
    \end{equation*}
    
    For $N=2n$,
    \begin{equation*}
        \lim_{\epsilon\to\infty} a_{n-1} = \lim_{\epsilon\to\infty} C_{n-1}a_n = \frac{2n - 3}{2n - 1} = \frac{N-3}{N-1}
    \end{equation*}
    
    For $N=2n+1$,
    \begin{equation*}
        \lim_{\epsilon\to\infty} a_{n-1} = \lim_{\epsilon\to\infty} C_{n-1}a_n = \frac{n-1}{n} = \frac{N-3}{N-1}
    \end{equation*}
    Therefore,
    \begin{equation*}
        \lim_{\epsilon\to\infty} 2\sum_{i=1}^{n}a_i^2p + \frac{(a_{n-1} + a_n)^2}{4} - a_{n-1} = \frac{1}{(N-1)^2}
    \end{equation*}
\end{proof}

\section{Optimization algorithm} \label{sec:optimization_algorithm}
\begin{itemize}
    \item For $N=2$, the optimal solution is easily derived with $p=\frac{1}{e^\epsilon + 1}$.
    \item For $N=3$, the process involves optimizing $p_0$ by solving the following equation:
    \begin{equation*}
        p_0 = \argmin_{p_0 \in [0, p]} \text{Var}[Y|x^*_1]
    \end{equation*}
\end{itemize}
In Algorithm \ref{alg:optimize}, we provide an optimization process designed to find an optimal N-output mechanism. The algorithm specifically addresses cases where $N\geq4$.

(Lines 1-21) The algorithm iteratively adjusts $N$, starting from $N=4$, to find the optimal set of parameters $\Omega$ and $p$ that minimize the worst-case noise variance. For convenience, the algorithm is designed to determine $p$ instead of $p_0$. (Line 2) An empty list $\mathcal{S}$ is initialized to store tuples of the form ($\Omega, p, max\text{Var}$). (Lines 4-5) For each iteration, use Theorem $\ref{thm:type0}$ with $p=\frac{1}{e^\epsilon +N-1}$ to obtain initial configuration $\Omega_0$. (Lines 6-7) If the condition $a_j < a_{j+1}$ does not hold for every $j$, it indicates that $N$ is too large, so the iteration stops. (Line 20) The optimal configuration is then found by selecting the tuple with the smallest $max\text{Var}$ from $\mathcal{S}$. (Lines 8-14) If the condition of Line 8 holds, Theorem \ref{thm:type0} provides the optimal configuration for the given $N$. (Lines 15-18) If the condition does not hold, Theorem \ref{thm:type1} is used to provide the optimal configuration for the given $N$.

\begin{algorithm} [t]
\caption{Optimization algorithm for minimizing the worst-case noise variance} \label{alg:optimize}
\begin{algorithmic}[1] 
\Require $\epsilon$
\Ensure $\Omega, p$
\State $N \gets 4$
\State $\mathcal{S} \gets []$
\While{True}
    \State $p = \frac{1}{e^\epsilon + N - 1}$
    \State Let $\Omega_0$ be the result from Theorem \ref{thm:type0} with $p$
    \If{there exists an $i$ such that $a_j >= a_{j+1}$}
        \State \textbf{break}
    \ElsIf{$Var_n[Y|x^*_n; \Omega_0, p] < Var_1[Y|x^*_1; \Omega_0, p]$}
        \If{$N = 2n$}
            \State $\mathcal{S}$.append(($\Omega_0$, $p$, $Var_n[Y|x^*_n; \Omega_0, p]$))
        \Else
            \State Let $\Omega_1$, $p_0$ be the result from \eqref{eq:opt_p0} and Theorem \ref{thm:type0}.
            \State $p \gets \frac{1 - p_0}{e^\epsilon + 2n -1}$
            \State $\mathcal{S}$.append(($\Omega_1$, $p$, $Var_n[Y|x^*_n; \Omega_1, p]$))
        \EndIf
    \Else
        \State $p \gets \frac{1}{e^\epsilon + N - 1}$
        \State Let $\Omega_2$ be the result from Theorem $\ref{thm:type1}$ with $p$
        \State $\mathcal{S}$.append(($\Omega_2$, $p$, $Var_n[Y|x^*_n; \Omega_1, p]$))
    \EndIf
    \State $N \gets N + 1$
\EndWhile
\State Find $\Omega$, $p$ that minimize $max$Var from $\mathcal{S}$

\State \Return $\Omega$, $p$
\end{algorithmic}
\end{algorithm}

\subsection{Derivation of probability transition matrix} \label{app:em_anal}
 Similar to $\Theta_{i, j}^{\mathcal{P}}$, we derive a general probability transition matrix for discrete case $\Theta_{i, j}^{\mathcal{N}}$ can be expressed as:
\begin{equation} \label{app:eq:theta_n}
    \Theta_{i, j}^{\mathcal{N}} = \frac{1}{l} \int_{R_i} \text{Pr}[y = a_j | x]  \ dx
\end{equation}
i) For $|j| \geq 1$, equation \eqref{app:eq:theta_n} can be represented as:
\begin{equation} \label{app:eq:theta_n_2}
\begin{split}
    \Theta_{i, j}^{\mathcal{N}} &= \frac{1}{l} \int_{R_i} p \ dx + \frac{1}{l} \int_{R_i} \text{Pr}[y=a_j | x] - p \ dx \\
    &= p + \frac{1}{l} \int_{R_i} \text{Pr}[y=a_j | x] - p \ dx
\end{split}
\end{equation}
\quad For $|j| \geq 2$, the second term of \eqref{app:eq:theta_n_2} is
\begin{equation*}
\begin{split}
    &\int_{R_i} \text{Pr}[y=a_j | x] - p \ dx \\
    &= \sum_{k=j}^{j+1} \mathbb{I}_{r_{i-1} \ \leq \  x_k \bigwedge r_i \ \geq \ x_{k-1}} \int_{\max \{ x_{k-1}, r_{i-1} \}} ^ { \min \{ x_k, r_i \} } p_{j,k}(x) - p \ dx
\end{split}
\end{equation*}
\quad For $j=-1$,  the second term of \eqref{app:eq:theta_n_2} is
\begin{equation*}
\begin{split}
    &\int_{R_i} \text{Pr}[y=a_{-1} | x] - p \ dx \\
    &= \sum_{k=-1}^1 \mathbb{I}_{r_{i-1} \ \leq \  x_{k} \bigwedge r_i \ \geq \ x_{k-1}} \int_{\max \{ x_{k-1}, r_{i-1} \}} ^ { \min \{ x_{k}, r_i \} } p_{j,k}(x) - p \ dx
\end{split}
\end{equation*}
\quad For $j=1$,  the second term of \eqref{app:eq:theta_n_2} is
\begin{equation*}
\begin{split}
    &\int_{R_i} \text{Pr}[y=a_{1} | x] - p \ dx \\
    &= \sum_{k=0}^2 \mathbb{I}_{r_{i-1} \ \leq \  x_{k} \bigwedge r_i \ \geq \ x_{k-1}} \int_{\max \{ x_{k-1}, r_{i-1} \}} ^ { \min \{ x_{k}, r_i \} } p_{j, k}(x) - p \ dx
\end{split}
\end{equation*}

ii) For $j=0$, equation \eqref{app:eq:theta_n} can be represented as:
\begin{equation*}
    \Theta_{i, j}^{\mathcal{N}} = p_0 + \frac{1}{l} \int_{R_i} \text{Pr}[y=a_0 | x] - p_0 \ dx
\end{equation*}
\quad where the second term is defined as:
\begin{equation*}
\begin{split}
    &\int_{R_i} \text{Pr}[y=a_{0} | x] - p_0 \ dx \\
    &= \sum_{k=0}^1 \mathbb{I}_{r_{i-1} \ \leq \  x_{k} \bigwedge r_i \ \geq \ x_{k-1}} \int_{\max \{ x_{k-1}, r_{i-1} \}} ^ { \min \{ x_{k}, r_i \} } p_{j,k}(x)-p_0 \ dx
\end{split}
\end{equation*}

\end{document}